\pdfoutput=1
\documentclass[12pt,reqno]{amsart}
\usepackage{etex}

\usepackage[latin1]{inputenc}
\usepackage{subfigure}
\usepackage{graphicx}

\usepackage{subfigure}
\usepackage{xcolor}           
\usepackage{enumitem}         
\usepackage{tikz}
\usetikzlibrary{fit,calc,positioning,decorations.pathreplacing,matrix, shapes}
\usetikzlibrary{trees}

\usepackage{todonotes}

\usepackage{relsize}

\usepackage{xspace}

\usepackage{amsfonts, amssymb, amsmath}
\usepackage{bm}               
\usepackage{stmaryrd} 

\newtheorem{thm}{Theorem}[section]
 \newtheorem{defi}[thm]{Definition}
 \newtheorem{exple}[thm]{Example}
 
  \newtheorem{prop}[thm]{Proposition}
  \newtheorem{lem}[thm]{Lemma}
 \newtheorem{cor}[thm]{Corollary}
  
   \newtheorem{rque}[thm]{Remark}

\usepackage{xargs}
\newcommand*\diff{\mathop{}\!\mathrm{d}}
\renewcommand{\S}{\operatorname{S_2}}
\newcommand{\SG}{\mathfrak{S}} 
\newcommand{\Vect}{\operatorname{Vect}}
\newcommand{\W}{\mathcal{W}}
\newcommandx{\WP}[1][1=p]{\{\W_1,\cdots,\W_{#1}\}} 
\renewcommand{\H}{\mathcal{H}}
\newcommandx{\HP}[1][1=p]{\{\H_1,\cdots,\H_{#1}\}} 
\newcommandx{\WiP}[2][1=i,2=p]{\{\W_{#1,1},\cdots,\W_{#1,#2}\}} 
\renewcommandx{\i}[2]{\llbracket #1,#2 \rrbracket}
\newcommand{\NAT}{\mathcal{N\!AT}} 
\newcommandx{\NATdk}[2][1=d,2=k]{\NAT_{#1,#2}} 
\newcommandx{\NATdkp}{\NAT_{d,k,\dir}} 
\renewcommand{\d}{d} 
\renewcommand{\k}{k} 
\newcommand{\B}{B} 
\newcommand{\M}{M} 
\renewcommand{\O}{O} 
\newcommand{\Bp}{\mathcal{B}_p} 
\newcommand{\Op}{\mathcal{O}_p} 
\newcommand{\BT}{\mathcal{BT}} 
\newcommand{\N}{T} 
\newcommand{\GFN}{\mathfrak{N}} 
\newcommand{\LGFN}{\mathfrak{G}} 
\newcommandx{\GFHdk}[2][1=d,2=k]{\mathfrak{H}_{#1,#2}} 
\newcommandx{\GFNdk}[2][1=d,2=k]{\mathfrak{N}_{#1,#2}} 
\newcommandx{\LGFNdk}[2][1=d,2=k]{\mathfrak{G}_{#1,#2}} 
\newcommand{\gs}{\wi_1 \times \dots \times \wi_d} 
\newcommand{\w}{w} 
\newcommand{\wi}{\w} 
\newcommand{\dir}{\pi} 
\newcommandx{\edir}[2][1=d,2=k]{\Pi_{#1,#2}} 
\newcommand{\V}{\mathcal{V}}
\newcommand{\F}{\mathcal{F}}
\newcommand{\LV}{\V_L} 
\newcommand{\RV}{\V_R} 
\newcommand{\inv}{\operatorname{Inv}} 
\newcommand{\imaj}{\operatorname{iMaj}} 
\newcommand{\std}{\operatorname{Std}} 
\newcommand{\DV}{\V_{\dir}} 
\newcommandx{\GFD}[1][1=d]{\mathfrak{D}_{#1}} 
\newcommandx{\DTp}[1][1=p]{\mathcal{DT}_{d,#1}} 
\newcommand{\p}{p} 
\newcommandx{\NATB}[1][1=\B]{\NAT\left(#1\right)} 
\newcommandx{\NATM}[1][1=\M]{\NAT\left(#1\right)} 
\newcommandx{\EL}{\mathcal{E}_L} 
\newcommandx{\ER}{\mathcal{E}_R} 
\newcommand{\E}{\mathcal{E}} 
\newcommand{\LO}{\mathcal{L}_0} 
\newcommand{\RO}{\mathcal{R}_0} 
\newcommandx{\nat}[2][1=d,2=k]{NAT$_{#1,#2}$\xspace}
\newcommand{\QQ}{\operatorname{\mathbb{Q}}}
\newcommand{\X}{X}

\newcommand{\BNAT}{\operatorname{\mathbb{M}}} 
\newcommand{\BSG}{\operatorname{\BNAT\SG}} 
\newcommand{\Bxy}{\operatorname{\mathbb{B}}} 

\newcommand{\Bx}{\operatorname{\mathbb{L}}} 

\newcommand{\pattern}{\begin{array}{c}\includegraphics[scale=1.0]{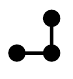}\end{array}}
\newcommand\sbullet[1][.5]{\mathbin{\vcenter{\hbox{\scalebox{#1}{$\bullet$}}}}}

\newcommand{\NodeBT}[2]{
\begin{tikzpicture}[baseline=(current bounding box.base)]
  \filldraw[color=black] (7pt, 7pt) circle (2pt);
  \draw (7pt,7pt)--(0,0);
  \node at (-3pt,-3pt) {$#1$};
  \draw (7pt,7pt)--(14pt,0pt);
  \node at (17pt,-3pt) {$#2$};
\end{tikzpicture}}

\renewcommand{\r}[1]{{\color{red}\ensuremath{{#1}}}}
\renewcommand{\b}[1]{{\color{blue}\ensuremath{{#1}}}}

\usepackage{todonotes}

\begin{document}
\author[J.-C. Aval]{Jean-Christophe Aval}
\address{Laboratoire Bordelais de Recherche en Informatique (UMR CNRS 5800),
Universit\'e de Bordeaux, 33405 TALENCE }
\author[al.]{Adrien Boussicault}
\address{Laboratoire Bordelais de Recherche en Informatique (UMR CNRS 5800),
Universit\'e de Bordeaux, 33405 TALENCE }
\author[]{B\'{e}r\'{e}nice Delcroix-Oger}
\address{Institut 
Math\'ematiques de Toulouse (UMR CNRS 5219),
Universit\'e Paul Sabatier, 31062 TOULOUSE}
\author[]{Florent Hivert}
\address{Laboratoire de Recherche en Informatique (UMR CNRS 8623),
B\^atiment 650, Universit\'e Paris Sud 11, 91405 ORSAY CEDEX}
\author[]{Patxi Laborde-Zubieta}
\address{Laboratoire Bordelais de Recherche en Informatique (UMR CNRS 5800),
Universit\'e de Bordeaux, 33405 TALENCE }

\title{Non-ambiguous trees: new results and generalisation}

\date{\today}


\begin{abstract}
We present a new definition of non-ambiguous trees (NATs) as labelled binary
trees. We thus get a differential equation whose solution can be described
combinatorially. This yields a new formula for the number of NATs. We also
obtain $q$-versions of our formula. We finally generalise NATs to higher
dimension. 
\end{abstract}

\maketitle
\section*{Introduction}

Non-ambiguous trees (NATs for short) were introduced in a previous paper \cite{AvaBouBouSil14}.
We propose in the present article a sequel to this work.

Tree-like tableaux \cite{ABN} are certain fillings of Ferrers diagram, in simple bijection 
with permutations or alternative tableaux \cite{postnikov,viennot}. They are the subject of an intense
research activity in combinatorics, mainly because they appear as the key tools
in the combinatorial interpretation of the well-studied model of statistical mechanics
called PASEP: they naturally encode the states of the PASEP, together with
the transition probabilities through simple statistics \cite{CorteelWilliams}.

Among tree-like tableaux, NATs were defined as rectangular-shaped objects in \cite{AvaBouBouSil14}.
In this way, they are in bijection with permutations $\sigma=\sigma_1\,\sigma_2\,\dots\,\sigma_n$
such that the excedances ($\sigma_i>i$) are placed at the beginning of the word $\sigma$.


Such permutations were studied by Ehrenborg and Steingrimsson \cite{ES},
who obtained an explicit enumeration formula. Thanks to NATs, a bijective proof
of this formula was described in \cite{AvaBouBouSil14}. 

In the present work, we define NATs as labelled binary trees (see Definition
\ref{def_nat}, which is equivalent to the original definition). This new
presentation allows us to obtain many new results about these objects. The plan
of the article is the following.

In Section \ref{nat}, we (re-)define NATs as
binary trees whose right and left children are respectively labelled with two
sets of labels. We show how the generating series for these objects satisfies
differential equations (Prop. \ref{diff_eq_nat}), whose solution is quite
simple and explicit (Prop. \ref{gen_ser_nat}). A combinatorial interpretation
of this expression involves the (new) notion of hooks in binary trees, linked
to the notion of leaves in ordered trees. Moreover this expression yields a new
formula for the number of NATs as a positive sum (see Proposition~\ref{thm_p}),
where Ehrenborg-Steingrimsson's formula is alternating. It should be noted that
Prop.~\ref{gen_ser_nat} and Proposition~\ref{thm_p} (in the case $\alpha=\beta=1)$
were already proven by Clark and Ehrenborg \cite{CE}. To conclude with Section
\ref{nat}, we obtain $q$-analogues of our formula, which are similar to those
obtained for binary trees in \cite{BjornerWachs,HivNovThib08} (see Theorem
\ref{thm-q-hook}, the relevant statistics are either the number of inversions
or the inverse major index).\\
Section \ref{gnat} presents a generalisation of
NATs in higher dimension. For any $k\le d$, we consider NATs of dimension
$(d,k)$, embedded  in $\mathbb{Z}^d$, and with edges of dimension $k$
\footnote{A definition in terms of labelled trees is given in Subsection
    \ref{def_natdk}.}. 
The original case corresponds to dimension $(2,1)$. Our main result on this
question is a differential equation satisfied by the generating series of these
new objects.

Finally, we study the (new) notion of hooks on binary trees in
Section \ref{hook}. We prove (through the use of generating series, and
bijectively) that the number of hooks is distributed on binary trees as another
statistics: the childleaf statistic, defined as the number of vertices who have
at least one leaf as a child.

\section{Non-ambiguous trees} \label{nat}

\subsection{Definitions} \label{definitions}

We recall that a \emph{binary tree} is a rooted tree whose vertices may have 
no child, one left child, one right child or both of them. 
The size of a binary tree is its number of vertices.
Usually, it is considered that there is a unique binary tree with
no vertex, it is called the \emph{empty binary tree}. In this article, we
consider that there are two binary trees of size 0: the \emph{left empty binary
tree} and the \emph{right empty binary tree}, they are respectively denoted by
$\emptyset_L$ and $\emptyset_R$. Having no child in the left direction (resp.
right direction) is the same as having the left (resp. right) empty subtree in
this direction. We denote by $\BT$ the set of binary trees.
Given a binary tree $B$, we denote by $\LV(B)$ and $\RV(B)$ the set of left
children (also called left vertices) and the set of right children (also called
right vertices). By convention, $\LV(\emptyset_L)=\RV(\emptyset_R)=-1$ and
$\LV(\emptyset_R)=\RV(\emptyset_L)=0$. We shall extend this notation to NATs.
Let $U$ and $V$ be two vertices of a binary tree $B$. If $V$ is a vertex of the
subtree of $B$ whose root is $U$, then $V$ is a \emph{descendant} of $U$ and $U$
an \emph{ancestor} of $V$.

We now define the notion of non-ambiguous trees:
\begin{defi} \label{def_nat}
    A \emph{non-ambiguous tree} (NAT) $T$ is a labelling of a binary tree $B$
    such that:
    \begin{itemize}
        \item the left (resp. right) children are labelled from $1$ to
            $|\LV(B)|$ (resp. $|\RV(B)|$), such that different left (resp.
            right) vertices have different labels. In other words, each left
            (right) label appears exactly once.
        \item if $U$ and $V$ are two left (resp. right) children in the tree,
            such that $U$ is an ancestor of $V$, then the label of $U$ in $T$ is
            strictly greater than the label of $V$.
    \end{itemize}

    The underlying binary tree of a non-ambiguous tree is called its \emph{shape}.
    By convention, there is a unique NAT whose shape is $\emptyset_L$ (resp.
    $\emptyset_R$) which is also denoted $\emptyset_L$ (resp. $\emptyset_R$).
    We denote by $\NAT(B)$ the set of NATs of shape $B$.
\end{defi}
It is sometimes useful to label the root as well. In
this case, it is considered as both a left and right child so that it carries
a pairs of labels, namely $(|\LV(T)|+1, |\RV(T)|+1)$. On pictures, to ease the
reading, we color the labels of left and right vertices in red and blue
respectively.

In \cite{AvaBouBouSil14}, NATs were defined in a slightly different, more geometrical way.
We recall it here, and show that the two definitions are equivalent.
Formally, a (geometric) NAT of size $n$ is a set $A$ of $n$ points $(x,y)\in\mathbb{N}\times\mathbb{N}$ such that:
\begin{enumerate}
\item \label{condition_1_ana} $(0,0)\in A$; we call this point the \emph{root} of $A$;
\item \label{condition_2_ana} given a non-root point $p=(x,y)\in A$, there exists one point $q=(x',y')\in A$ such that $y'<y$ and $x'=x$, or one point $r=(x',y')\in A$ such that $x'<x$, $y'=y$, but not both (which means that the pattern $\pattern$ 
is avoided);
\item \label{condition_3_ana} there is no empty line between two given points: if there exists a point $p=(x,y)\in A$, then for every $x'<x$ (resp. $y'<y$) there exists $q=(x'',y'')\in A$ such that $x''=x'$ (resp. $y''=y'$).
\end{enumerate}

\begin{lem}
This geometric definition is equivalent to Definition 
\ref{def_nat}.
\end{lem}

\begin{proof}
Let us consider a NAT $T$ presented as a labelled tree 
(Definition \ref{def_nat}). 
Each vertex of $T$ may be given a pair of coordinates $
(x,y)$ as follows.
For a left (resp. right) child, its $x$ (resp. $y$) 
coordinate is its label in $T$.
and its $y$ (resp. $x$) coordinate is the one of its 
closest ancestor which is a right child, or the root.
We thus get a geometrical object, which can be drawn as 
shown in Figure \ref{fig_exple_nat} (top, right).
With the minor change of coordinates $(\tilde{x},
\tilde{y})=(|\LV(T)|+1-x, |\RV(T)|+1-y)$, we get a set of 
points which
satisfies the geometrical definition.
The only thing that needs to be checked is the avoidance 
of the pattern
.
Let us proceed by reductio ad absurdum. We suppose that 
there are three vertices $p=(\tilde{x},\tilde{y})$, 
$q=(\tilde{x'},\tilde{y})$ and $r=(\tilde{x},\tilde{y'})$, 
with $\tilde{x'}<\tilde{x}$ and $\tilde{y'}<\tilde{y}$. We 
can suppose moreover that $p$ is a right child and that 
its parent is then $r$. Thus, $q$ is not an ancestor of $p
$. Let us consider the closest ancestor $c$ of $q$ which 
is a right vertex. Then the label of $c$ in $T$ is the 
same as the one of $p$, but it is not $p$ as $\tilde{x'}<
\tilde{x}$. It is then absurd that $p$ and $c$ have the 
same label in $T$.

Conversely, to go from the geometric version of a NAT to 
Definition \ref{def_nat}, 
we just have to forget the redundant coordinate.
\end{proof}
The top part of Figure \ref{fig_exple_nat} shows an 
example of a NAT,
and illustrates the correspondence between the geometrical 
presentation
of \cite{AvaBouBouSil14} and Definition \ref{def_nat}.

The dimension $\wi_L(T)\times \wi_R(T)$ of the rectangle containing the
geometrical presentation of $T$, is called the geometric size of $T$ and
satisfies
\[
    (\wi_L(T),\wi_R(T)) = (|\LV(T)|+1, |\RV(T)|+1).
\]

By convention, the non-ambiguous trees $\emptyset_L$ and $\emptyset_R$ satisfy
respectively $(\wi_L(\emptyset_L),\wi_R(\emptyset_L)) = (0,1)$ and $(\wi_L(\emptyset_R),\wi_R(\emptyset_R)) = (1,0)$.

\tikzset{every node/.style={inner sep=1pt,scale=0.8}}
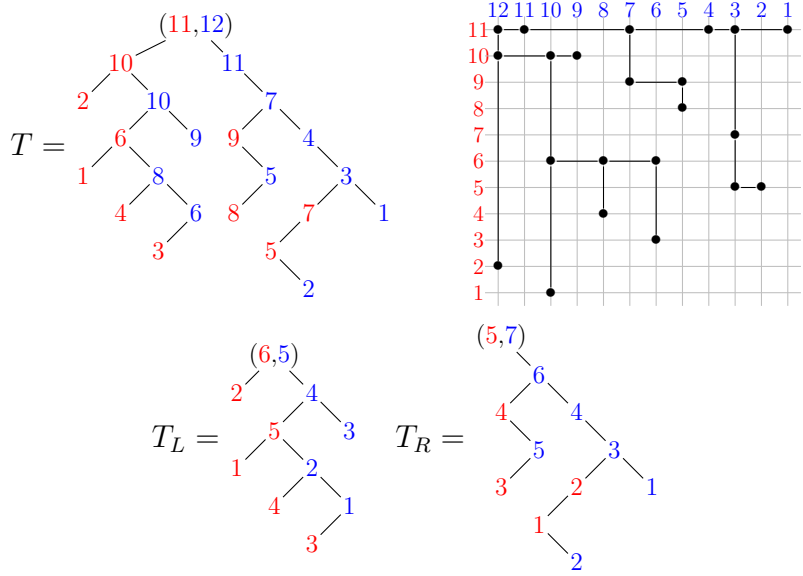
\begin{figure} 
\begin{center}
\[T=\begin{tikzpicture}[baseline=(current bounding box.center),scale=0.5]
\node (r) at (0,8) {(\textcolor{red}{11},\textcolor{blue}{12})};
\node (11b) at (1,7) {\textcolor{blue}{11}};
\node (10b) at (-1,6) {\textcolor{blue}{10}};
\node (9b) at (0,5) {\textcolor{blue}{9}};
\node (8b) at (-1,4) {\textcolor{blue}{8}};
\node (7b) at (2,6) {\textcolor{blue}{7}};
\node (6b) at (0,3) {\textcolor{blue}{6}};
\node (5b) at (2,4){\textcolor{blue}{5}};
\node (4b) at (3,5) {\textcolor{blue}{4}};
\node (3b) at (4,4){\textcolor{blue}{3}};
\node (2b) at (3,1){\textcolor{blue}{2}};
\node (1b) at (5,3){\textcolor{blue}{1}};
\node (10r) at (-2,7){\textcolor{red}{10}};
\node (9r) at (1,5){\textcolor{red}{9}};
\node (8r) at (1,3){\textcolor{red}{8}};
\node (7r) at (3,3){\textcolor{red}{7}};
\node (6r) at (-2,5){\textcolor{red}{6}};
\node (5r) at (2,2){\textcolor{red}{5}};
\node (4r) at (-2,3){\textcolor{red}{4}};
\node (3r) at (-1,2){\textcolor{red}{3}};
\node (2r) at (-3,6){\textcolor{red}{2}};
\node (1r) at (-3,4){\textcolor{red}{1}};
\draw (2r)--(10r)--(r)--(11b)--(7b)--(4b)--(3b)--(1b);
\draw (7b)--(9r);
\draw (9r)--(5b);
\draw (5b)--(8r);
\draw (3b)--(7r);
\draw (7r)--(5r);
\draw (5r)--(2b);
\draw (10r)--(10b);
\draw (9b)--(10b)--(6r)--(1r);
\draw (6b)--(8b)--(4r);
\draw (3r)--(6b);
\draw (6r)--(8b);
\end{tikzpicture}
\hspace*{1cm}
\begin{tikzpicture}[baseline=(current bounding box.center),
                    scale=0.35,
                    every node/.style={inner sep=0pt,scale=0.7}]
\draw (0,0.75) node{\textcolor{blue}{12}};
\draw (1,0.75) node{\textcolor{blue}{11}};
\draw (2,0.75) node{\textcolor{blue}{10}};
\draw (3,0.75) node{\textcolor{blue}{9}};
\draw (4,0.75) node{\textcolor{blue}{8}};
\draw (5,0.75) node{\textcolor{blue}{7}};
\draw (6,0.75) node{\textcolor{blue}{6}};
\draw (7,0.75) node{\textcolor{blue}{5}};
\draw (8,0.75) node{\textcolor{blue}{4}};
\draw (9,0.75) node{\textcolor{blue}{3}};
\draw (10,0.75) node{\textcolor{blue}{2}};
\draw (11,0.75) node{\textcolor{blue}{1}};
\draw (-0.75,0) node{\textcolor{red}{11}};
\draw (-0.75,-1) node{\textcolor{red}{10}};
\draw (-0.75,-2) node{\textcolor{red}{9}};
\draw (-0.75,-3) node{\textcolor{red}{8}};
\draw (-0.75,-4) node{\textcolor{red}{7}};
\draw (-0.75,-5) node{\textcolor{red}{6}};
\draw (-0.75,-6) node{\textcolor{red}{5}};
\draw (-0.75,-7) node{\textcolor{red}{4}};
\draw (-0.75,-8) node{\textcolor{red}{3}};
\draw (-0.75,-9) node{\textcolor{red}{2}};
\draw (-0.75,-10) node{\textcolor{red}{1}};
\draw[gray!50] (0,0.5)--(0,-10.5);
\draw[gray!50] (1,0.5)--(1,-10.5);
\draw[gray!50] (2,0.5)--(2,-10.5);
\draw[gray!50] (3,0.5)--(3,-10.5);
\draw[gray!50] (4,0.5)--(4,-10.5);
\draw[gray!50] (5,0.5)--(5,-10.5);
\draw[gray!50] (6,0.5)--(6,-10.5);
\draw[gray!50] (7,0.5)--(7,-10.5);
\draw[gray!50] (8,0.5)--(8,-10.5);
\draw[gray!50] (9,0.5)--(9,-10.5);
\draw[gray!50] (10,0.5)--(10,-10.5);
\draw[gray!50] (11,0.5)--(11,-10.5);
\draw[gray!50] (-0.5,0)--(11.5,0);
\draw[gray!50] (-0.5,-1)--(11.5,-1);
\draw[gray!50] (-0.5,-2)--(11.5,-2);
\draw[gray!50] (-0.5,-3)--(11.5,-3);
\draw[gray!50] (-0.5,-4)--(11.5,-4);
\draw[gray!50] (-0.5,-5)--(11.5,-5);
\draw[gray!50] (-0.5,-6)--(11.5,-6);
\draw[gray!50] (-0.5,-7)--(11.5,-7);
\draw[gray!50] (-0.5,-8)--(11.5,-8);
\draw[gray!50] (-0.5,-9)--(11.5,-9);
\draw[gray!50] (-0.5,-10)--(11.5,-10);
\node (r) at (0,0) {$\bullet$};
\node (11b) at (12-11,0) {$\bullet$};
\node (10b) at (12-10,10-11) {$\bullet$};
\node (9b) at (12-9,10-11) {$\bullet$};
\node (8b) at (12-8,6-11) {$\bullet$};
\node (7b) at (12-7,0) {$\bullet$};
\node (6b) at (12-6,6-11) {$\bullet$};
\node (5b) at (12-5,9-11){$\bullet$};
\node (4b) at (12-4,0) {$\bullet$};
\node (3b) at (12-3,0){$\bullet$};
\node (2b) at (12-2,5-11){$\bullet$};
\node (1b) at (12-1,0){$\bullet$};
\node (10r) at (0,10-11){$\bullet$};
\node (9r) at (12-7,9-11){$\bullet$};
\node (8r) at (12-5,8-11){$\bullet$};
\node (7r) at (12-3,7-11){$\bullet$};
\node (6r) at (12-10,6-11){$\bullet$};
\node (5r) at (12-3,5-11){$\bullet$};
\node (4r) at (12-8,4-11){$\bullet$};
\node (3r) at (12-6,3-11){$\bullet$};
\node (2r) at (0,2-11){$\bullet$};
\node (1r) at (12-10,1-11){$\bullet$};
\draw (2r)--(10r)--(r)--(11b)--(7b)--(4b)--(3b)--(1b);
\draw (7b)--(9r);
\draw (9r)--(5b);
\draw (5b)--(8r);
\draw (3b)--(7r);
\draw (7r)--(5r);
\draw (5r)--(2b);
\draw (10r)--(10b);
\draw (9b)--(10b)--(6r)--(1r);
\draw (6b)--(8b)--(4r);
\draw (3r)--(6b);
\draw (6r)--(8b);
\end{tikzpicture}\]
\[T_L=
\begin{tikzpicture}[baseline=(current bounding box.center),scale=0.5]
\node (10b) at (-1,6) {\textcolor{blue}{4}};
\node (9b) at (0,5) {\textcolor{blue}{3}};
\node (8b) at (-1,4) {\textcolor{blue}{2}};
\node (6b) at (0,3) {\textcolor{blue}{1}};
\node (10r) at (-2,7){(\textcolor{red}{6},\textcolor{blue}{5})};
\node (6r) at (-2,5){\textcolor{red}{5}};
\node (4r) at (-2,3){\textcolor{red}{4}};
\node (3r) at (-1,2){\textcolor{red}{3}};
\node (2r) at (-3,6){\textcolor{red}{2}};
\node (1r) at (-3,4){\textcolor{red}{1}};
\draw (2r)--(10r);
\draw (10r)--(10b);
\draw (9b)--(10b)--(6r)--(1r);
\draw (6b)--(8b)--(4r);
\draw (3r)--(6b);
\draw (6r)--(8b);
\end{tikzpicture}
\hspace{0.5cm}
T_R=
\begin{tikzpicture}[baseline=(current bounding box.center),scale=0.5]
\node (11b) at (1,7) {(\textcolor{red}{5},\textcolor{blue}{7})};
\node (7b) at (2,6) {\textcolor{blue}{6}};
\node (5b) at (2,4){\textcolor{blue}{5}};
\node (4b) at (3,5) {\textcolor{blue}{4}};
\node (3b) at (4,4){\textcolor{blue}{3}};
\node (2b) at (3,1){\textcolor{blue}{2}};
\node (1b) at (5,3){\textcolor{blue}{1}};
\node (9r) at (1,5){\textcolor{red}{4}};
\node (8r) at (1,3){\textcolor{red}{3}};
\node (7r) at (3,3){\textcolor{red}{2}};
\node (5r) at (2,2){\textcolor{red}{1}};
\draw (11b)--(7b)--(4b)--(3b)--(1b);
\draw (7b)--(9r);
\draw (9r)--(5b);
\draw (5b)--(8r);
\draw (3b)--(7r);
\draw (7r)--(5r);
\draw (5r)--(2b);
\end{tikzpicture}\]
\caption{A non-ambiguous tree, its geometrical presentation,  
and its left and right subtrees}\label{fig_exple_nat}
\end{center}
\end{figure}

Figure \ref{fig_exple_nat_B} gives an example of a class $NAT(B)$,
in the case where the binary tree is $B=\begin{tikzpicture}[baseline=(current bounding box.center),scale=0.2]
\node (r) at (0,2) {$\sbullet$};
\node (1b) at (2,1) {$\sbullet$};
\node (2b) at (-0.7,0) {$\sbullet$};
\node (1r) at (-2,1){$\sbullet$};
\node (2r) at (0.7,0){$\sbullet$};
\draw (2r)--(1b)--(r)--(1r)--(2b);
\end{tikzpicture}
$.

\tikzset{every node/.style={inner sep=1pt,scale=0.8}}
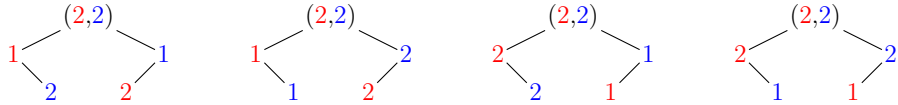
\begin{figure} 
\begin{center}
\[
\begin{tikzpicture}[baseline=(current bounding box.center),scale=0.5]
\node (r) at (0,2) {(\textcolor{red}{2},\textcolor{blue}{2})};
\node (1b) at (2,1) {\textcolor{blue}{1}};
\node (2b) at (-1,0) {\textcolor{blue}{2}};
\node (1r) at (-2,1){\textcolor{red}{1}};
\node (2r) at (1,0){\textcolor{red}{2}};
\draw (2r)--(1b)--(r)--(1r)--(2b);
\end{tikzpicture}
\hspace*{1cm}
\begin{tikzpicture}[baseline=(current bounding box.center),scale=0.5]
\node (r) at (0,2) {(\textcolor{red}{2},\textcolor{blue}{2})};
\node (2b) at (2,1) {\textcolor{blue}{2}};
\node (1b) at (-1,0) {\textcolor{blue}{1}};
\node (1r) at (-2,1){\textcolor{red}{1}};
\node (2r) at (1,0){\textcolor{red}{2}};
\draw (2r)--(2b)--(r)--(1r)--(1b);
\end{tikzpicture}
\hspace*{1cm}
\begin{tikzpicture}[baseline=(current bounding box.center),scale=0.5]
\node (r) at (0,2) {(\textcolor{red}{2},\textcolor{blue}{2})};
\node (1b) at (2,1) {\textcolor{blue}{1}};
\node (2b) at (-1,0) {\textcolor{blue}{2}};
\node (2r) at (-2,1){\textcolor{red}{2}};
\node (1r) at (1,0){\textcolor{red}{1}};
\draw (1r)--(1b)--(r)--(2r)--(2b);
\end{tikzpicture}
\hspace*{1cm}
\begin{tikzpicture}[baseline=(current bounding box.center),scale=0.5]
\node (r) at (0,2) {(\textcolor{red}{2},\textcolor{blue}{2})};
\node (2b) at (2,1) {\textcolor{blue}{2}};
\node (1b) at (-1,0) {\textcolor{blue}{1}};
\node (2r) at (-2,1){\textcolor{red}{2}};
\node (1r) at (1,0){\textcolor{red}{1}};
\draw (1r)--(2b)--(r)--(2r)--(1b);
\end{tikzpicture}
\]
\caption{An example of $NAT(B)$}\label{fig_exple_nat_B}
\end{center}
\end{figure}

\subsection{Differential equations on non-ambiguous trees}
\label{differential}

The goal of this section is to get (new) formulas for the number of NATs 
with prescribed shape.

The crucial argument is the following remark: let $T$ be a NAT
whose shape is a non-empty binary tree $B = \NodeBT{L}{R}$.
Restricting the labellings of the left and right children of $T$
to $L$ and $R$ gives non-decreasing labelling of their respective left and
right children. Note that the root of $L$ (resp. $R$) is a left
(resp. right) child in $T$. By renumbering the labels so that they are
consecutive numbers starting from $1$, we get two non-ambiguous labellings for
$L$ and $R$, that is two non-ambiguous trees $T_L$ and
$T_R$. See Figure~\ref{fig_exple_nat} for an example. 

\begin{rque}
The geometric size of $T$ satisfies:
\begin{align}
    \wi_L(T) &= \wi_L(T_L)+\wi_L(T_R) \\
    \wi_R(T) &= \wi_R(T_L)+\wi_R(T_R) 
\end{align}
\end{rque}

Conversely, knowing the labelling of $L$ and $R$, to recover the labelling
of $T$, one has to choose which labels among $\{ 1, \ldots, \LV(T) \}$ will be used for
$L$ (including its root) and the same for right labels.  
As a consequence, the number of NAT with shape $B$ is given by:

\begin{equation}\label{Equation:BNat-binom}
  \left|\NATB[\NodeBT{L}{R}]\right| =
  \binom{|\LV(T)|}{|\LV(R)|}
  \binom{|\RV(T)|}{|\RV(L)|}\,
  |\NATB[L]|\,
  |\NATB[R]|.
\end{equation}
Our first step is to recover the hook-length formula for the number of
NATs of fixed shape (\cite{AvaBouBouSil14}). We
use the method from~\cite{HivNovThib08}, namely, applying recursively a
bilinear integro-differential operator called here a \emph{pumping function}
along a binary tree.

First of all, we consider the $\QQ$-vector space $\QQ\NAT$ of formal sums of
non-ambiguous trees and identifies $\NATB$ with the formal sum of its elements.
We consider also the $\QQ$-vector spaces $\QQ\NAT_L$ and $\QQ\NAT_R$ generated
respectively by $\NAT\setminus\{\emptyset_R\}$ and
$\NAT\setminus\{\emptyset_L\}$. Let $\BNAT$ be the linear map
\[
    \BNAT:\QQ\NAT_L\times\QQ\NAT_R\mapsto\QQ\NAT
\]
sending a pair of non-ambiguous trees $(T_1,T_2)$ to the formal sum of NATs $T$
such that $T_L=T_1$ and $T_R=T_2$. The main remark is that $\NATB$ can be
computed by a simple recursion using $\BNAT$.
\begin{lem}\label{lem:recur_B}
    The formal sum $\NATB[B]$ of non-ambiguous trees of shape $B$ satisfies
    the following recursion: if $B=\emptyset_L$ or $B=\emptyset_R$ then
    $\NATB[B] = B$, else
    \[
        \NATB[\NodeBT{L}{R}] = \BNAT\left(\NATB[L], \NATB[R]\right).
    \]
\end{lem}
To count non-ambiguous trees, and as suggested by the binomial coefficients
in~\eqref{Equation:BNat-binom}, we shall use \emph{doubly exponential
generating functions} in two variables $x$ and $y$, where $x$ and $y$ count the
geometric size: the weight of a NAT $T$ is $\Phi(T):=
\frac{x^{\wi_L(T)}}{\wi_L(T)!} \frac{y^{\wi_R(T)}}{\wi_R(T)!}\,.$ We extend
$\Phi(T)$ by linearity to a map $\QQ\NAT\mapsto\QQ[[x, y]]$. Consequently,
$\Phi(\NATB[B])$ is the generating series of the non-ambiguous trees of shape
$B$. Thanks to~\eqref{Equation:BNat-binom} the image in $\QQ[[x, y]]$ of the
bilinear map $\BNAT$ under the map $\Phi$ is a simple differential operator:
\begin{defi}\label{def:pumping}
  The \emph{pumping function} $\Bxy$ is the bilinear map
  $\QQ[[x,y]]\times\QQ[[x,y]]\mapsto\QQ[[x,y]]$ defined by
  \begin{equation}
    \Bxy(f, g) = \int_0^x\int_0^y
    \partial_yf(u,v)\cdot
    \partial_xg(u,v)
    \diff u\diff v.
  \end{equation}
  We further define recursively, for any binary tree $B$
  an element $\Bx(B)\in\QQ[[x,y]]$ by
  \begin{equation}
  \Bx(\emptyset_L) = y,
  \quad
  \Bx(\emptyset_R) = x
  \quad\text{and}\qquad
  \Bx\left({\NodeBT{L}{R}}\right) = \Bxy\left(\Bx(L), \Bx(R)\right)\,.
\end{equation}
\end{defi}
A simple computation gives us that 
for $f=\frac{x^{a_1}}{a_1!}\frac{y^{b_1}}{b_1!}$ and $g=\frac{x^{a_2}}{a_2!} \frac{y^{b_2}}{b_2!}$
one has 
$$
\Bxy(f,g)=\binom{a_1+a_2-1}{a_2-1} \binom{b_1+b_2-1}{b_1-1} \frac{x^{a_1+a_2}}{(a_1+a_2)!} \frac{y^{b_1+b_2}}{(b_1+b_2)!}.
$$
Whence for $(T_1, T_2)$ a pair of NATs in $\QQ\NAT_L\times\QQ\NAT_R$, we get:
\begin{align*}
\Bxy(\Phi(T_1),\Phi(T_2)) =& \binom{w_L(T_1)+w_L(T_2)-1}{w_L(T_2)-1} \binom{w_R(T_1)+w_R(T_2)-1}{w_R(T_1)-1}\\
&\times \frac{x^{w_L(T_1)+w_L(T_2)}}{(w_L(T_1)+w_L(T_2))!} \frac{y^{w_R(T_1)+w_R(T_2)}}{(w_R(T_1)+w_R(T_2))!} 
\end{align*}

Thus~\eqref{Equation:BNat-binom} implies the following lemma.
\begin{lem}\label{lem:com_phi}
For $(T_1, T_2)$ a pair of NATs in $\QQ\NAT_L\times\QQ\NAT_R$, one has
  $$\Phi(\BNAT(T_1, T_2)) = \Bxy(\Phi(T_1),\Phi(T_2)).$$
\end{lem}

We derive from these two lemmas the following proposition.
\begin{prop}
\label{prop:commutation_phi}
 For any non-empty binary tree $B$,
  $$\Phi(\NATB[B])=\Bx(B).$$
\end{prop}
\begin{proof}
We apply Lemmas \ref{lem:recur_B} and \ref{lem:com_phi} to write (with $L$ and $R$ the left and right sub-tree of $B$):
\begin{align*}
	\Phi(\NATB[B])
	&=
	\Phi(\BNAT(\NAT(L),\NAT(R)))
	\\
	& = 
	\Phi(\BNAT(
		\sum_{T_1 \in \NAT(L)} T_1,
		\sum_{T_2 \in \NAT(R)} T_2
	))
	\\
	&=
	\sum_{T_1 \in \NAT(L)} 
	\ 
	\sum_{T_2 \in \NAT(R)} 
	\Phi(\BNAT(
		T_1, T_2
	))
	\\
	&=
	\sum_{T_1 \in \NAT(L)} 
	\ 
  \sum_{T_2 \in \NAT(R)} 
	\Bxy( 
		\Phi(T_1), \Phi(T_2)
	)
	\\
	&=
	\Bxy\left( 
		\sum_{T_1 \in \NAT(L)} \Phi(T_1),
		\sum_{T_2 \in \NAT(R)} \Phi(T_2)
	\right)
	=
	\Bxy( \Bx(L), \Bx(R) ).
\end{align*}
\end{proof}

We are now able to recover the hook-length formula
of~\cite{AvaBouBouSil14} for non-ambiguous trees of a given shape.
\begin{prop} \label{prop_hook_nat} 
    Let $\B$ be a binary tree. For each non-root left (resp. right) vertex $U$, we
    denote by $\EL(U)$ (resp. $\ER(U)$) the number of left (resp. right) vertices of
    the subtree with root $U$ (itself included in the count). Then
    \begin{equation}\label{eq:HLT}
    |\NAT(\B)|
    =
    \frac{ |\LV(\B)|! \cdot |\RV(\B)|!}{
      \displaystyle
      \prod_{U: \text{left child}}\EL(U) \cdot 
      \prod_{U: \text{right child}}\ER(U)
    }\,.
    \end{equation}
\end{prop}

Before proving this Proposition, let us illustrate it on an example.
\begin{exple} \label{exple_hook}
  Let $B = \scalebox{0.5}
{ \newcommand{\nodea}{\node[draw,fill,circle] (a) {$$};}
  \newcommand{\nodeb}{\node[draw,fill,circle] (b) {$$};}
  \newcommand{\nodec}{\node[draw,fill,circle] (c) {$$};}
  \newcommand{\noded}{\node[draw,fill,circle] (d) {$$};}
  \newcommand{\nodee}{\node[draw,fill,circle] (e) {$$};}
  \newcommand{\nodef}{\node[draw,fill,circle] (f) {$$};}
  \newcommand{\nodeg}{\node[draw,fill,circle] (g) {$$};}
  \newcommand{\nodeh}{\node[draw,fill,circle] (h) {$$};}
\begin{tikzpicture}[scale=0.5,baseline=(current bounding box.center),
                      every node/.style={inner sep=2pt}]
\matrix[column sep=.15cm, row sep=.15cm,ampersand replacement=\&]{
         \&         \&         \& \nodea  \&         \&         \&         \&         \&         \\
         \& \nodeb  \&         \&         \&         \&         \&         \& \nodee  \&         \\
 \nodec  \&         \& \noded  \&         \&         \& \nodef  \&         \&         \& \nodeh  \\
         \&         \&         \&         \&         \&         \& \nodeg  \&         \&         \\
};
\path[thick] (b) edge (c) edge (d)
    (f) edge (g)
    (e) edge (f)
    (a) edge (b) edge (e)
        (e) edge (h);
\end{tikzpicture}}$. The hook formula is given by:
    \begin{equation}
    |\NAT(\B)|
    =
    \frac{ 3! 4!}{
      \left( 1 \cdot 2 \cdot 1\right) \cdot
      \left(1 \cdot 1 \cdot 3 \cdot 1 \right)
    }=24\,.
    \end{equation}

\end{exple}

\begin{proof}
Proposition \ref{prop:commutation_phi} may be rewritten as:
\begin{equation}\label{eq:com_phi}
\Bx(B) = |\NAT(\B)| \frac{x^{\wi_L(B)}y^{\wi_R(B)}}{\wi_L(B)! \wi_R(B)!}.
\end{equation}

So, we get by a simple computation : 
\begin{align*}
\Bx(B)
&=
	\Bx\left({\NodeBT{L}{R}}\right)
	=
	\Bxy\left(\Bx(L), \Bx(R)\right)
\\
&= 
	\Bxy\left(
		|\NAT(L)|
		\frac{
			x^{\wi_L(L)}y^{\wi_R(L)} 
		}{
			\wi_L(L)! \wi_R(L)!
		}
		,
		|\NAT(R)|
		\frac{
			x^{\wi_L(R)}y^{\wi_R(R)}
		}{
			\wi_L(R)! \wi_R(R)!
		}
	\right)
\\
&=
	\frac{
		|\NAT(L)| \cdot |\NAT(R)|
	}{
		\wi_L(L)! \wi_R(L)! \wi_L(R)! \wi_R(R)!
	}
	\Bxy\left(
		x^{\wi_L(L)}y^{\wi_R(L)},
		x^{\wi_L(R)}y^{\wi_R(R)}
	\right)
\\
&=
	\frac{
		|\NAT(L)| \cdot |\NAT(R)| \wi_R(L) \wi_L(R)
	}{
		\wi_L(L)! \wi_R(L)! \wi_L(R)! \wi_R(R)!
	}
	\frac{
		x^{\wi_L(L)+\wi_L(R)}y^{\wi_R(L)+\wi_R(R)}
	}{
		(\wi_L(L)+\wi_L(R)) \cdot (\wi_R(L)+\wi_R(R))
	}
\\
&=
	\frac{
		 \wi_R(L) \wi_L(R)
	}{
		\wi_L(B) \wi_R(B)
	}
	|\NAT(L)| \frac{x^{\wi_L(L)}y^{\wi_R(L)}}{\wi_L(L)! \wi_R(L)!}
	\cdot 	
	|\NAT(R)| \frac{x^{\wi_L(R)}y^{\wi_R(R)}}{\wi_L(R)! \wi_R(R)!}
\end{align*}
which, together with \eqref{eq:com_phi} gives:
\begin{align*}
|\NAT(\B)| \frac{x^{\wi_L(B)}y^{\wi_R(B)}}{\wi_L(B)! \wi_R(B)!}
&=
	\frac{
		 \wi_R(L) \wi_L(R)
	}{
		\wi_L(B) \wi_R(B)
	}
	\cdot 	
	|\NAT(L)| \frac{x^{\wi_L(L)}y^{\wi_R(L)}}{\wi_L(L)! \wi_R(L)!}
	\\
	& \cdot
	|\NAT(R)| \frac{x^{\wi_L(R)}y^{\wi_R(R)}}{\wi_L(R)! \wi_R(R)!}\,.
\end{align*}

We obtain
\begin{align*}
\frac{
	|\NAT(\B)| 
}{
	(\wi_L(B)-1)! (\wi_R(B)-1)!
}
&=
	\wi_R(L) \wi_L(R)
	\frac{
		|\NAT(L)| \cdot |\NAT(R)| 
	}{
		\wi_L(L)! \wi_R(L)! \wi_L(R)! \wi_R(R)!
	}
\\
&=
	\frac{
		1
	}{
		\wi_L(L) \wi_R(R)
	}
	\frac{
		|\NAT(L)|
	}{
		(\wi_L(L)-1)! (\wi_R(L)-1)!
	}
	\\
	&\cdot
	\frac{
		|\NAT(R)| 
	}{
		(\wi_L(R)-1)! (\wi_R(R)-1)!
	}\,.
\end{align*}

We deduce that
\begin{align*}
\frac{
	|\NAT(\B)| 
}{
	\LV(B)! \RV(B)!
}
&=
	\frac{
		1
	}{
		\wi_L(L) \wi_R(R)
	}
	\cdot
	\frac{
		|\NAT(L)|
	}{
		\LV(L)! \RV(L)!
	}
	\cdot
	\frac{
		|\NAT(R)| 
	}{
		\LV(R)! \RV(R)!
	}\,.
\end{align*}

The coefficient $\frac{1}{\wi_L(L) \wi_R(R)}$ recursively gives the denominator of the hook-legth formula, and we get \eqref{eq:HLT}.
\end{proof}
\bigskip

Let $\LGFN$ be the \emph{exponential generating function of non-ambiguous trees}
with weight $\Phi$:
\begin{equation}
\LGFN(x,y):= \sum_{\N \in \NAT} \Phi(\N) = 
  \sum_{\N \in \NAT} \frac{x^{\wi_L(T)}}{\wi_L(T)!} \frac{y^{\wi_R(T)}}{\wi_R(T)!}\,.
\end{equation}
and $\GFN$ its derivative $\GFN=\partial_x\partial_y\LGFN$. Naturally,
they are linked by the relation
\[
    \LGFN(x,y) = y + x + \int_0^x \int_0^y \GFN(u,v)\diff u\diff v.
\]

They are both solutions of a fixed point differential equation.
\begin{prop}
\label{prop_equ_diff_nat}
The generating function $\GFN$ and $\LGFN$ can be computed by the following
fixed point differential equations:
\begin{equation}
\label{equ_gfn}
\LGFN = y + x + \int_x \int_y \partial_x \LGFN \cdot \partial_y \LGFN 
\hspace{.5cm}
\text{and}
\hspace{.5cm}
\GFN = 
\left( 1 + \int_{x} \GFN  \right)
\cdot
\left( 1 + \int_{y} \GFN  \right)
\end{equation}
\end{prop}

\begin{proof}
The first equation is a consequence of the definition of the bilinear
map $\Bxy$:
\begin{align*}
\LGFN   &= \sum_{B \in \BT} \Bx(B)\\
        &= y + x + \sum_{(L, R) \in \BT_L\times\BT_R} \Bx\left({\NodeBT{L}{R}}\right)\\
        &= y + x + \sum_{(L, R) \in \BT_L\times\BT_R} \Bxy(\Bx(L), \Bx(R))\\
        &= y + x + \Bxy(\LGFN-x, \LGFN-y)\\
        &= y + x + \Bxy(\LGFN, \LGFN),
\end{align*}
with $\BT_L=\BT\setminus\{\emptyset_R\}$ and
$\BT_R=\BT\setminus\{\emptyset_L\}$. To prove the second equation, remark that
the first equation implies the identity
\[
    \partial_x \partial_y \LGFN = \partial_x \LGFN. \partial_y \LGFN
\]
and moreover that we have
\[
    \partial_x \LGFN = 1 + \int_y \GFN
    \hspace{.5cm}
    \text{and}
    \hspace{.5cm}
    \partial_y \LGFN = 1 + \int_x \GFN.\
\]
\end{proof}

From these identities, a closed formula can be computed for $\GFN$ and $\LGFN$.
The expression of $\GFN$ was already proven in \cite{CE} using permutations.

\begin{prop}\label{prop:gf_expressions}
The doubly exponential generating functions for non-ambiguous trees are given by
$$
\LGFN = x+y-\log( 1 - (e^x-1)(e^y-1) )\,,
\quad\text{and}\quad
\GFN
=
\frac{
e^{x+y}
}{
\left(
1 - (e^x - 1)(e^y - 1)
\right)^2
}.
$$
\end{prop}

\begin{proof}
We know that $\LGFN$ is a solution of
\begin{equation}
\label{equ_lgfn}
        \begin{cases}
            \partial_x \partial_y f = \partial_x f \times \partial_y f,\\
            f(x,y)=f(y,x)
        \end{cases}
\end{equation}

This system of equations satisfies the two following properties:
\begin{itemize}
    \item if $s$ is a solution of Equation~\ref{equ_lgfn} then for each power
        series $\varphi$ with constant term equal to zero,
        $f(\varphi(x),\varphi(y))$ is also a solution;
    \item if we fix the initial condition $f(x,0)$, there exists a unique
        formal power series solution to Equation~\ref{equ_lgfn}.
\end{itemize}

Let $f$ be a particular solution. Let us consider the notation
$f_x:=\partial_xf$ and $f_y:=\partial_yf$, then
\[
    \partial_xf_y\cdot\partial_yf_x=f_x^2f_y^2.
\]
We suppose that $\partial_xf_y=f_y^2$ and $\partial_yf_x=f_x^2$, hence
\[
    f_y=\frac{-1}{x+c_1(y)} \text{ and } f_x=\frac{-1}{y+c_2(x)}.
\]
Since $f_x^2=\partial_y f_x=f_xf_y$, we get $f_y=f_x$, which implies
\[
    x+c_1(y)=y+c_2(x).
\]
As a consequence, $c_1(z)=c_2(z)=z+c$ with $c$ a real number.
Finally
\[
    f(x,y)=-\ln(x+y+c).
\]

Conversely, $-\ln(x+y+c)$ satisfies Equation~\eqref{equ_lgfn}. It remains to
find a real number $c$ and formal power series $\varphi$ such that
$\LGFN(x,0)=-\ln(\varphi(x)+c)$. Since $\LGFN(x,0)=x$, we get
$\varphi(x)=e^{-x}-c$. Moreover, the condition $\varphi(0)=0$ implies $c=1$.
As a consequence
\[
    \LGFN(x,y)=-\ln(e^{-x}+e^{-y}-1),
\]
which can be rewritten as
\[
    \LGFN(x,y)=y+x-\ln(1-(e^x-1)(e^y-1)).
\]
Differentiating with respect to $x$ and $y$, we find the expression of $\GFN$.
\end{proof}

In the context of the PASEP \cite{CorteelWilliams}, it is natural to consider statistics $\LO(\B)$ and $\RO(\B)$. 
Let us briefly recall that the partially asymmetric exclusion process  (PASEP) is a model for a system 
of interacting particles hopping left and right on a one-dimensional lattice of $n$ sites. 
In the general case, the probability of a given state of the model depends on five parameters $(q,\alpha,\beta,\gamma,\delta)$
which give the probability of transitions (a particle moving to the right or the left, or going in or out the model, when possible). 
Tree-like tableaux \cite{ABN} have been proven to give a combinatorial interpretation for the steady state of the PASEP when $\gamma=\delta=0$,
the weight in $\alpha$ (resp. $\beta$) corresponding to the statistic $\LO(\B)$ (resp. $\RO(\B)$) defined as follows.

\begin{defi}
The \emph{leftmost branch} of a binary tree $\B$ is the set of vertices
$\{s_0,\ldots,s_k\}$ such that $s_0$ is the root of $\B$, $s_k$ is a leaf and
$s_{i+1}$ is the left child of $s_i$, for each $i<k$. Similarly, we define the \emph{rightmost
branch} of a binary tree.
We denote by $\LO(\B)$ and $\RO(\B)$ the number of non-root
vertices respectively in the leftmost and rightmost branches.
\end{defi}
 We extend these
definitions to non-ambiguous trees. For example, in
Figure~\ref{fig_exple_nat}, we have $\LO(T)=2$ and $\RO(T)=5$. 
Let us define the following $(\alpha, \beta)$-generating function for
non-ambiguous trees:
\[
    \GFN(x,y;\alpha, \beta)
    =
    \sum_{\N \in \NAT}
        \frac{
            x^{|\LV(\N)|}
            \cdot
            y^{|\RV(\N)|}
            \cdot
            \alpha^{\LO(\N)}
            \cdot
            \beta^{\RO(\N)}
        }{
            |\LV(\N)|! \cdot |\RV(\N)|!
        }.
\]
It satisfies an $(\alpha,\beta)$-analogue for the identity of
Proposition~\ref{prop_equ_diff_nat}.
\begin{prop}\label{diff_eq_nat}
A differential equation for $\GFN(x,y;\alpha, \beta)$ is
$$
        \GFN(x,y;\alpha, \beta)=
        \left(
            1 + \alpha \int_x \GFN(u,y;\alpha, 1)\diff u
        \right)
        \cdot
        \left(
            1 + \beta \int_y \GFN(x,v;1, \beta)\diff v
        \right).
$$
\end{prop}
\begin{proof}
We just need to define a new pumping function:
\[
    \Bxy^{(\alpha, \beta)}(f,g)
    =
    \alpha\beta \Bxy(\left.f\right|_{\beta=1},\left.g\right|_{\alpha=1})
\]
and deduce the expected differential equation.
\end{proof}

The solution of the new differential equation is given by 
Proposition~\ref{pro_equ_diff_alpha_beta}, a bijective proof is given in
Section~\ref{bijectionslabelledtree}.
\begin{prop}\label{gen_ser_nat}
The $(\alpha,\beta)$-exponential generating function for non-ambiguous trees is equal to
\label{pro_equ_diff_alpha_beta}
\[
    \GFN(x,y;\alpha, \beta)
    =
    \frac{
        e^{\alpha x + \beta y}
    }{
        \left(
            1 - (e^x -1)(e^y - 1)
        \right)^{\alpha+\beta}
    }.
\]
\end{prop}
If we develop this expression we obtain an $(\alpha,\beta)$-analogue of the
enumeration of non-ambiguous trees of fixed geometric size. 

Let us recall that the $q$-analogue of the Stirling
numbers of the second kind $\operatorname S_{2,q}(n,k)$
is the number of ways to partition a set of size $n$ in $k$ parts, where the power of $q$ counts the number of elements different from $n$ in the subset containing $n$. For instance, $S_{2,q}(3,2) = 1+2q$, $S_{2,q}(4,2) = 1+3q+3q^2$ and $S_{2,q}(4,3) =3+3q$.

\begin{prop}\label{pro_ana_enum_alpha_beta} \label{thm_p}
    Let $i$ and $j$ be two positive integers. The $(\alpha,\beta)$-analogue of
    the number of $NATs$ of geometric size $i\times j$ is
    \[
        \left[\frac{x^{i-1}\,y^{j-1}}{(i-1)!\,(j-1)!}\right]
        \GFN(x,y;\alpha,\beta)
        =
        \sum_{p\geqslant 1}(p-1)!\,(\alpha+\beta)_{p-1}
        \operatorname S_{2,\alpha}(i,p)\operatorname S_{2,\beta}(j,p).
    \]
    with $q^{(n)}:=q(q+1)\cdots(q+n-1)$
the rising factorial and $\operatorname S_{2,q}(n,k)$ the 
$q$-analogue of the Stirling
numbers of the second kind. 
\end{prop}

\subsection{Combinatorial interpretation with the zigzag bijection of Burstein}
\label{bijectionslabelledtree}

The purpose of this subsection is to explain combinatorially Propositions~\ref{gen_ser_nat} and \ref{pro_ana_enum_alpha_beta}.
To do so, we use the ``zigzag'' bijection introduced and studied in \cite{SteiWill}. This bijection, that we will denote by $\varphi$, was further studied by Burstein \cite{Bur07}.

First, let us introduce the statistic that corresponds to the integer $p$ in the
enumeration formula of Proposition~\ref{pro_ana_enum_alpha_beta}.
\begin{defi}\label{def_hook}
    Let $B$ be a binary tree and $v$ one of its node. The \emph{hook} of a
    vertex $v$ is the union of $\{v\}$, its leftmost branch and its rightmost
    branch. We say that $v$ is the root of its hook. There is a unique way to
    partition the vertices into hooks. The number of hooks in such a partition is
    the \emph{hook number of the tree} and it is denoted by $hook(T)$. We extend this
    definition to non-ambiguous trees.
\end{defi}

\begin{rque}
We can obtain recursively the unique partition of a binary tree into hooks by
extracting the root's hook and iterating the process on each tree of the
remaining forest.
\end{rque}

\begin{exple} In Figure \ref{fig:fig_nat}, the hook whose root is \textcolor{blue}{10} is highlighted in purple. The partition of all the vertices into hooks can be seen by
  keeping only bold edges. The roots of the hooks are $\{(\textcolor{red}{11},\textcolor{blue}{12}), \textcolor{blue}{10}, \textcolor{red}{9}, \textcolor{blue}{8}, \textcolor{red}{8}, \textcolor{red}{7}, \textcolor{red}{3}, \textcolor{blue}{2}\}$, and so the hook number of the tree is $8$.
\begin{figure} 
\begin{center}
\begin{tikzpicture}[scale=0.45]
\node (r) at (0,8) {(\textcolor{red}{11},\textcolor{blue}{12})};
\node (11b) at (1,7) {\textcolor{blue}{11}};
\node[fill=purple!50] (10b) at (-1,6) {\textcolor{blue}{10}};
\node[fill=purple!50] (9b) at (0,5) {\textcolor{blue}{9}};
\node (8b) at (-1,4) {\textcolor{blue}{8}};
\node (7b) at (2,6) {\textcolor{blue}{7}};
\node (6b) at (0,3) {\textcolor{blue}{6}};
\node (5b) at (2,4){\textcolor{blue}{5}};
\node (4b) at (3,5) {\textcolor{blue}{4}};
\node (3b) at (4,4){\textcolor{blue}{3}};
\node (2b) at (3,1){\textcolor{blue}{2}};
\node (1b) at (5,3){\textcolor{blue}{1}};
\node (10r) at (-2,7){\textcolor{red}{10}};
\node (9r) at (1,5){\textcolor{red}{9}};
\node (8r) at (1,3){\textcolor{red}{8}};
\node (7r) at (3,3){\textcolor{red}{7}};
\node[fill=purple!50] (6r) at (-2,5){\textcolor{red}{6}};
\node (5r) at (2,2){\textcolor{red}{5}};
\node (4r) at (-2,3){\textcolor{red}{4}};
\node (3r) at (-1,2){\textcolor{red}{3}};
\node (2r) at (-3,6){\textcolor{red}{2}};
\node[fill=purple!50] (1r) at (-3,4){\textcolor{red}{1}};
\draw[line width=3pt] (2r)--(10r)--(r)--(11b)--(7b)--(4b)--(3b)--(1b);
\draw  (7b)--(9r);
\draw[line width=3pt] (9r)--(5b);
\draw (5b)--(8r);
\draw (3b)--(7r);
\draw[line width=3pt] (7r)--(5r);
\draw (5r)--(2b);
\draw (10r)--(10b);
\draw[line width=3pt, purple!50] (9b)--(10b)--(6r)--(1r);
\draw[line width=3pt] (6b)--(8b)--(4r);
\draw (3r)--(6b);
\draw (6r)--(8b);
\end{tikzpicture}
\caption{Hooks (bold edges) on a non-ambiguous tree}\label{fig:fig_nat}
\end{center}
\end{figure}
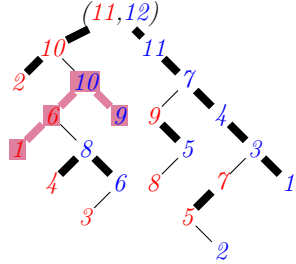
\end{exple}
The correspondence between $p$ and this new statistic is proven hereafter.

Let us now define the bijection $\varphi$ between non-ambiguous trees and permutations. From now on, we juggle between the geometric representation and the labelled binary
tree representation of non-ambiguous trees. Let $T$
be a non-ambiguous tree. We remove the first column of $T$, and denote by $T'$ the result of this deletion. We number, starting with
1, the south-east border, starting from the westmost edge
(Figure~\ref{fig:ana_zigzag}). Let $\sigma$ be the permutation $\varphi(T)$ and
$i$ the positive integer corresponding to a border edge. The image $\sigma(i)$
is defined as follows. Let $e$ be the border edge numbered by $i$ and suppose
that $e$ is vertical. If $e$ has no point to its left in the same row in $T'$, then
$\sigma(i)=i$. Else, starting from the leftmost point of the row in $T'$, we go down to
the closest point in the same column, then right to the closest point in the
same row and so on by doing a "zigzag" which alternates going down and right, until we reach a border edge $e'$. The image $\sigma(i)$
corresponds to the integer associated to $e'$. If $e$ is horizontal, we start
with the topmost point of the same column and then we "zigzag", starting from
with right direction and going alternatively down and right, to find $\sigma(i)$. For example, if $T$ if the
non-ambiguous tree of Figure~\ref{fig:ana_zigzag}, then $\sigma(23)=13$ and
$\sigma(3)=7$, and more generally
\[
    \sigma=(13\;1\;6\;20\;12\;5\;22\;10\;2\;23)\;(21)\;(18\;3\;7\;17\;15\;4\;19)\;(14\;9\;8\;16)\;(11)
\]
where we use the cyclic notation for permutations.

\begin{rque}\cite[Section 2]{Bur07}\label{rem:BurCyc}
It is a simple consequence of the construction that every cycle of $\varphi(T)$ corresponds either to an empty line or to a binary tree in $T'$. 
\end{rque}
\begin{figure}[h]
    \begin{center}
        \includegraphics[scale=.6]{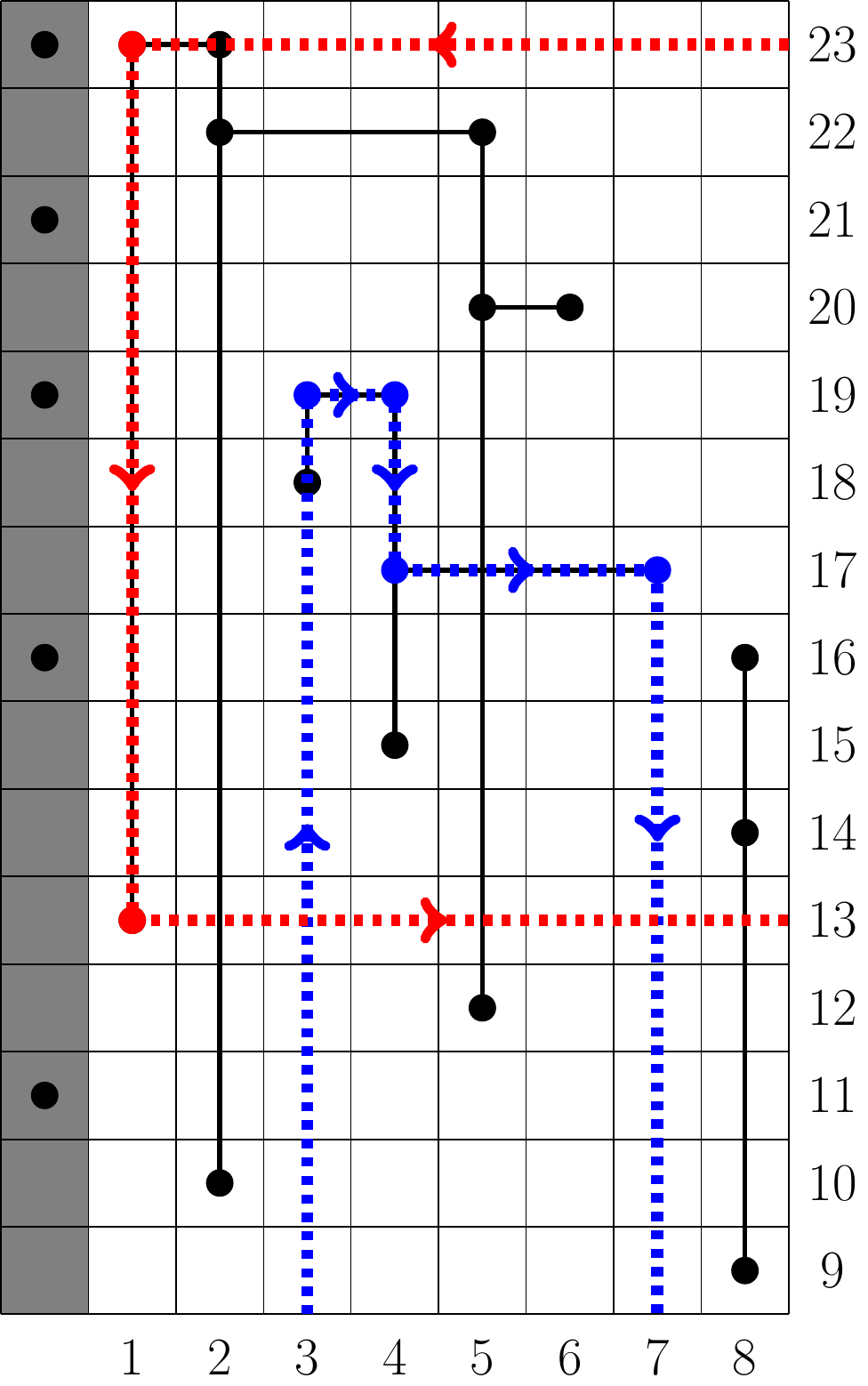}
        \caption{Example of the bijection of Burstein. }
        \label{fig:ana_zigzag}
    \end{center}
\end{figure}

Let us recall that an excedance in a permutation $\sigma$ is an entry $i$ such that $\sigma(i)>i$. In the cyclic notation, it corresponds to an ascent (an integer followed -cyclically- by a greater one).
\begin{prop}\cite[Theorem 14]{SteiWill}
    Let $\w_L$ and $\w_R$ be two positive integers. The map $\varphi$ is a
    bijection between non-ambiguous trees of geometric size $\w_L\times\w_R$ and
    permutations of size $\w_L+\w_R-1$ such that all their excedances are at
    positions $1,\cdots,\w_R-1$.
\end{prop}
If we do the same construction without deleting the first column of $T$ (in this case the corresponding horizontal border edge is numbered with $0$), we get another map, which we denote $\psi$. As observed in Remark~\ref{rem:BurCyc}, $\psi(T)$ is a cyclic permutation.

\begin{figure}[h]
    \begin{center}
        \includegraphics[scale=.6]{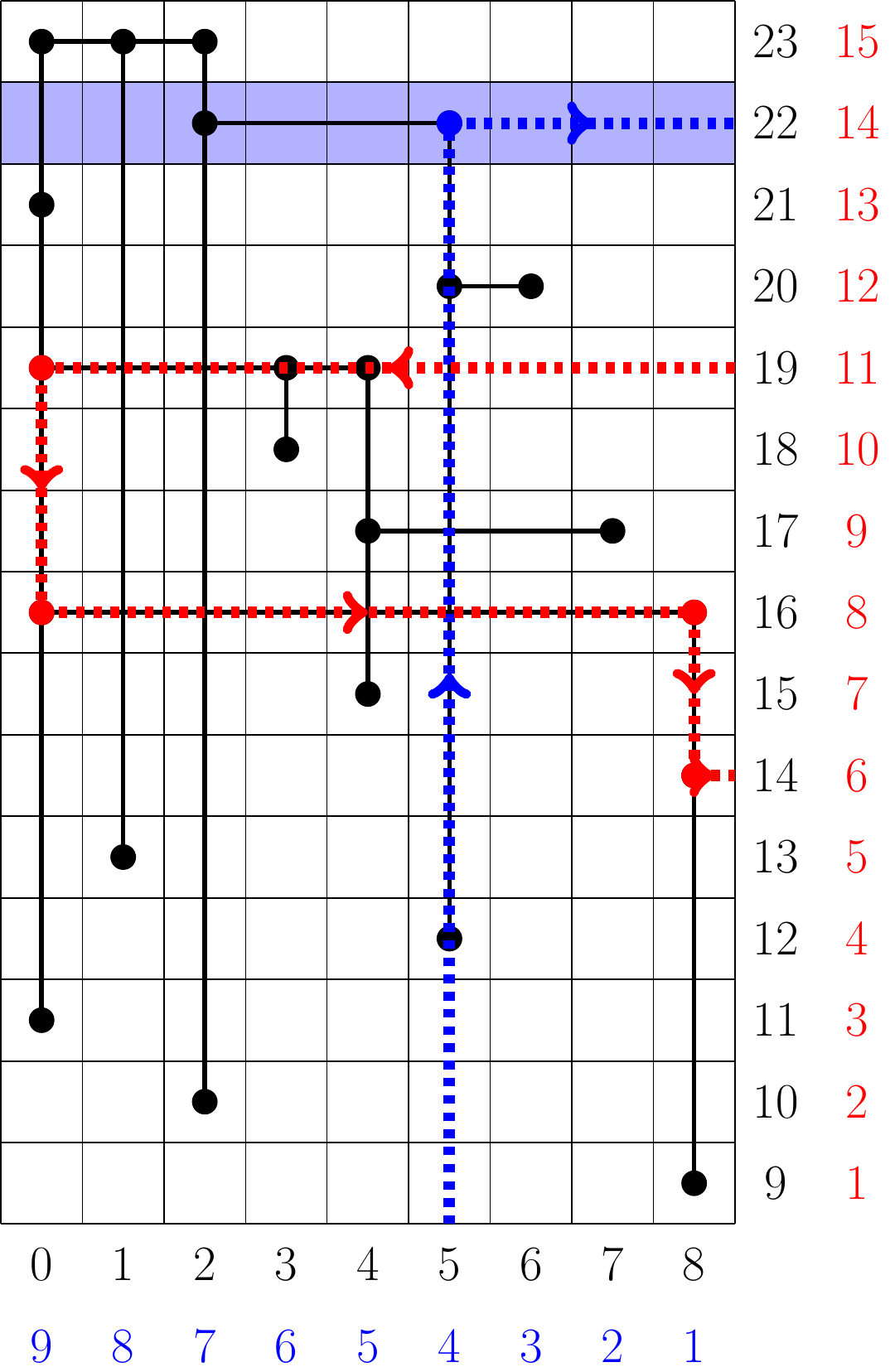}
        \caption{Example of the bijection $\psi$.}
        \label{fig:ana_zigzag_cycle}
    \end{center}
\end{figure}

The link between $\varphi$ and $\psi$ is given by the following lemma.
\begin{lem}\label{lem:ana_perm_cycles_cycle_bleus_rouges}
    Let $T$ be a non-ambiguous tree. Let $c_1\cdots c_k$ be the decomposition in
    cycles of $\varphi(T)$ such that the biggest element of $c_i$
    is larger than the biggest element of $c_{i+1}$.
    We denote by $m_i$ the representative of $c_i$ ending by its biggest element. 
    Then, a representative of $\psi(T)$ is the word $0m_1\cdots m_k$.
\end{lem}
This is illustrated on the running example of Figure \ref{fig:ana_zigzag_cycle}. On this non-ambiguous tree, denoted by $T$, $\psi(T)$ is the cycle
\[
     \psi(T)=\left(
                0\;
                13\;1\;6\;20\;12\;5\;22\;10\;2\;23\;
                21\;
                18\;3\;7\;17\;15\;4\;19\;
                14\;9\;8\;16\;
                11
            \right).
\]
Recall that $\varphi(T)$ is given by
\[
    \varphi(T)=(13\;1\;6\;20\;12\;5\;22\;10\;2\;23)\;(21)\;(18\;3\;7\;17\;15\;4\;19)\;(14\;9\;8\;16)\;(11)
\]
In particular, $\psi(19)=14=\varphi(16)$.

\begin{proof}
Let us denote $\sigma=\varphi(T)$ and $c = \psi(T)$. 
It is clear that only the images of entries that have a point in the leftmost column of $T$ (the one deleted in $T'$) are modified.
Let us denote by $i_1 < \cdots < i_k$ these entries.
The construction implies that $c(i_j ) = \sigma(i_{j-1})$ for $1<j\le k$, $c(i_1) = 0$ and $c(0) = \sigma(i_k)$. Whence the result.
\end{proof}

This lemma has two consequences. 
The first one is that, as $\varphi$ is a bijection, $\psi$ is also one.

The second consequence is that the excedances of $\psi(T)$ are precisely those of $\varphi(T)$, together with $0$.
\begin{cor}\label{cor:ana_bij_cycle}
    Let $\w_L$ and $\w_R$ be two positive integers. The map $\psi$ is a
    bijection between non-ambiguous trees of geometric size $\w_L\times\w_R$ and
    cycles of $\llbracket 0,\w_L+\w_R\rrbracket$ such that all their excedances are at
    positions $0,\cdots,\w_R-1$.
\end{cor}
In order to ease future explanations, we number independently rows and columns.
We replace the integers $\llbracket 0,\w_R-1\rrbracket$ with $\llbracket
\b{1},\b{\w_R}\rrbracket$ by using the map $i\mapsto \b{(\w_R-i)}$, and the
integers $\llbracket \w_R,\w_R+\w_L-1\rrbracket$ with $\llbracket
\r{1},\r{\w_L}\rrbracket$ by using the map $i\mapsto \r{(i-\w_R+1)}$, as shown
in Figure~\ref{fig:ana_zigzag_cycle}. If we denote $T$ the non-ambiguous tree of
this figure, then $\psi(T)$ is equal to the cycle
\[
    \left(
        \b{9}\;\r{5}\;\b{8}\;\b{3}\;\r{12}\;\r{4}\;\b{4}\;\r{14}\;
        \r{2}\;\b{7}\;\r{15}\;\r{13}\;\r{10}\;\b{6}\;\b{2}\;\r{9}\;
        \r{7}\;\b{5}\;\r{11}\;\r{6}\;\r{1}\;\b{1}\;\r{8}\;\r{3}
    \right).
\]
\begin{defi}\label{def:ana_cycle_bi}
    Let $i$ and $j$ be two positive integers. We define a \emph{2-coloured block
    decreasing cycle of size $i\times j$} as a cycle of the set
    $\llbracket \r{1},\r{i} \rrbracket \cup \llbracket \b{1},\b{j} \rrbracket$
    such that, if the image of an element $a$ is an element $b$ of the same
    colour then $a>b$.
\end{defi}
Using this definition, the map $\psi$ is a bijection between non-ambiguous trees
of geometric size $i\times j$ and 2-coloured block decreasing cycles of size
$i\times j$. Moreover, the number of blue blocks of the 2-coloured block
decreasing cycles has a simple interpretation over non-ambiguous trees.
\begin{lem}\label{lem:ana_enumeration_nb_bloc_fixe}
    Let $T$ be a non-ambiguous tree, then the hook number of $T$ is equal to the
    number of blue block in $\psi(T)$.
\end{lem}

For example, if $T$ is the non-ambiguous tree of
Figure~\ref{fig:ana_zigzag_cycle}, its hook number is 7, which is also the number of
blue blocks in $\psi(T)$.

\begin{proof}
    To prove this lemma, we show that, on the geometrical interpretation of a non-ambiguous tree, the vertical border edges on which ends a zigzag path starting from an horizontal border edge are exactly the vertical border edges of a line containing the right branch of a hook. For instance, on Figure~\ref{fig:ana_zigzag_cycle}, these vertical border edges are these numbered by $\textcolor{red}{15}$, $\textcolor{red}{14}$, $\textcolor{red}{12}$, $\textcolor{red}{11}$, $\textcolor{red}{9}$, $\textcolor{red}{8}$ and $\textcolor{red}{5}$.
    
    A zigzag starting from an horizontal border edge first reach a node on the right branch of a hook (as it has no vertices above). Then every vertical step reaches either the root of a hook or an horizontal border edge, while every horizontal  step reaches either a non-root node on the right branch of a hook or a vertical border edge. The only reached nodes are then on the right branch of a hook or its root. If such a zigzag path ends on a vertical border edge, the last step is horizontal from the root of a hook or a node on its right branch: the corresponding line contains the right branch of a node.
    
    Similarly, a zigzag starting from a vertical border edge first reach a node on the left branch of a hook (as it has no vertices on its left). Then every horizontal step reaches either the root of a hook or a vertical border edge, while every vertical step reaches either a non-root point on the left branch of a hook or an horizontal border edge. The only reached nodes are then on the left branch of a hook or its root. The root of a hook in this case can only be reached after an horizontal step. If such a zigzag path ends on a vertical border edge, the last step is then horizontal from a non-root node on the left branch of a hook: the corresponding line does not contain the right branch of a node, which proves the result.
\end{proof}

 From this lemma, we deduce the following proposition.
\begin{prop}
    Let $i, j$ and $p$ be positive integers. The number of non-ambiguous
    trees of geometric size $i\times j$ and hook number $p$ is
    \[
        (p-1)!\,p!\,\S(i,p)\,\S(j,p).
    \]

    Moreover, the doubly exponential generating series of non-empty
    non-ambiguous trees,  with weight on a NAT $T$ given by
    \[
        z^{hook(T)}
        \frac{
            x^{\w_L(T)}\,y^{\w_R(T)}
        }{
            \w_L(T)!\,\w_R(T)!
        },
    \]
    is
    \[
        -\ln\left(1-z(e^x-1)(e^y-1)\right).
    \]
\end{prop}
\begin{proof}
    According to Corollary~\ref{cor:ana_bij_cycle} and
    Lemma~\ref{lem:ana_enumeration_nb_bloc_fixe}, non-ambiguous trees with geometric size $i \times j$ with hook number $p$ is in bijection with 2-coloured
    blocks decreasing cycles of size $i\times j$ with $p$ blue blocks, which are counted by the enumeration formula.

    The generating series is the one corresponding to a cycle whose elements are pairs formed by a non-empty blue set and a non-empty red set. Recalling that the generating series of cycles is $-\ln(1-u)$ and the one of pairs of sets is     $z(e^x-1)(e^y-1)$, we get the result by composition.
\end{proof}
We extend naturally this proposition with the parameters $\alpha$ and $\beta$.
\begin{thm}\label{th:ana_zigzag_alpha_beta}
    Let $i, j$ and $p$ be positive integers. The number of non-ambiguous
    trees of geometric size $i\times j$ and hook number $p$ is
    \[
        (p-1)!\,(\alpha+\beta)_{p-1}\,\operatorname S_{ 2,\alpha }(i,p)\,\operatorname S_{ 2,\beta }(j,p).
    \]

    Moreover, the doubly exponential generating series of non-empty
    non-ambiguous trees,  with weight on a NAT $T$ given by
    \[
        \alpha^{\LO(T)}
        \beta^{\RO(T)}
        z^{hook(T)}
        \frac{
            x^{\w_L(T)-1}\,y^{\w_R(T)-1}
        }{
            (\w_L(T)-1)!\,(\w_R(T)-1)!
        },
    \]
    is
    \[
        \frac{
            ze^{\alpha x+\beta y}
        }{
            \left(1-z(e^x-1)(e^y-1)\right)^{\alpha+\beta}
        }.
    \]
\end{thm}
\begin{proof}
We first prove the expression $\frac{
            ze^{\alpha x+\beta y}
        }{
            \left(1-z(e^x-1)(e^y-1)\right)^{\alpha+\beta}
        }$. Let us first remark that this can be rewritten as $z \times e^{\alpha x+\beta y} \times e^{-\alpha\ln\left(1-z(e^x-1)(e^y-1)\right)} \times e^{-\beta\ln\left(1-z(e^x-1)(e^y-1)\right)}$. 
    This expression is the translation in terms of generating function of the following decomposition. A non-ambiguous tree can be decomposed into three pieces:
    \begin{itemize}
        \item the set of child-free  (i.e. with no child outside of the hook) vertices in the hook $h$ of the root (in yellow on figure \ref{fig:fig_nat2})
        \item the set of non-ambiguous trees attached to the left branch of $h$ (in green on figure \ref{fig:fig_nat2})
        \item and the set of non-ambiguous trees attached to the right branch of
            $h$ (in orange on figure \ref{fig:fig_nat2}).
    \end{itemize}
    
    \begin{figure} 
\begin{center}
\begin{tikzpicture}[scale=0.45]
\node (r) at (0,0) {(\textcolor{red}{15},\textcolor{blue}{9})};
\node[fill=orange!50,below right =1cm of r] (8b)  {\textcolor{blue}{8}};
\node[fill=orange!50,below right =1cm of 8b] (7b)  {\textcolor{blue}{7}};
\node[fill=orange!50,below left =0.5cm of 7b] (14r) {\textcolor{red}{14}};
\node[fill=orange!50,below right =0.5cm of 14r] (4b) {\textcolor{blue}{4}};
\node[fill=orange!50,below left =0.5cm of 4b] (12r) {\textcolor{red}{12}};
\node[fill=orange!50,below right =0.5cm of 12r] (3b) {\textcolor{blue}{3}};
\node[fill=yellow!50,below left =0.5cm of r] (13r) {\textcolor{red}{13}};
\node[fill=green!50,below left =1cm of 13r] (11r) {\textcolor{red}{11}};
\node[fill=green!50,below left =1cm of 11r] (8r) {\textcolor{red}{8}};
\node[fill=green!50,below right =0.5cm of 8r] (1b) {\textcolor{blue}{1}};
\node[fill=green!50, below right =0.5cm of 11r] (6b)  {\textcolor{blue}{6}};
\node[fill=green!50,below right =0.5cm of 6b] (5b) {\textcolor{blue}{5}};
\node[fill=green!50,below left =0.5cm of 5b] (9r) {\textcolor{red}{9}};
\node[fill=green!50,below right =0.5cm of 9r] (2b) {\textcolor{blue}{2}};
\node[fill=green!50,below left =0.5cm of 6b] (10r) {\textcolor{red}{10}};
\node[fill=green!50,below left =0.5cm of 9r] (7r) {\textcolor{red}{7}};
\node[fill=green!50,below left =0.5cm of 1b] (6r) {\textcolor{red}{6}};
\node[fill=orange!50,below left =0.5cm of 8b] (5r){\textcolor{red}{5}};
\node[fill=orange!50,below left =0.5cm of 12r] (4r){\textcolor{red}{4}};
\node[fill=yellow!50,below left =0.5cm of 8r] (3r){\textcolor{red}{3}};
\node[fill=orange!50,below left =0.5cm of 14r] (2r) {\textcolor{red}{2}};
\node[fill=green!50,below left =0.5cm of 6r] (1r) {\textcolor{red}{1}};
\draw[line width=3pt] (3r)--(8r)--(11r)--(13r)--(r)--(8b)--(7b);
\draw  (8b)--(5r);
\draw  (7b)--(14r);
\draw[line width=3pt] (2r)--(14r)--(4b);
\draw  (4b)--(12r);
\draw[line width=3pt] (4r)--(12r)--(3b);
\draw (11r)--(6b);
\draw[line width=3pt] (10r)--(6b)--(5b);
\draw (5b)--(9r);
\draw[line width=3pt] (7r)--(9r)--(2b);
\draw (8r)--(1b);
\draw[line width=3pt] (1b)--(6r)--(1r);
\end{tikzpicture}
        \includegraphics[scale=0.6]{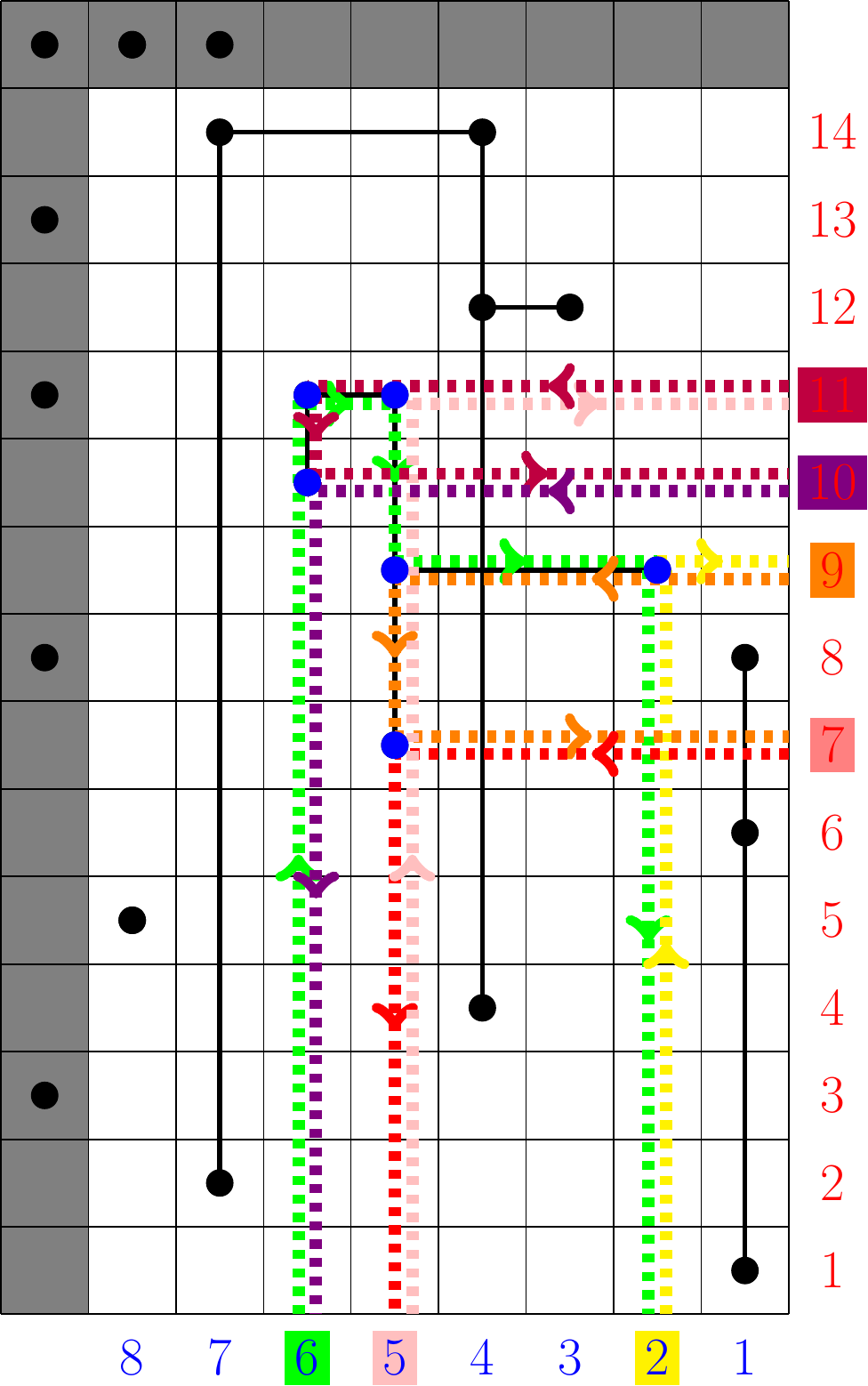}
\caption{Illustration of the proof of Theorem \ref{th:ana_zigzag_alpha_beta} with both tree and geometrical representation of the same NAT}\label{fig:fig_nat2}
\end{center}
\end{figure}
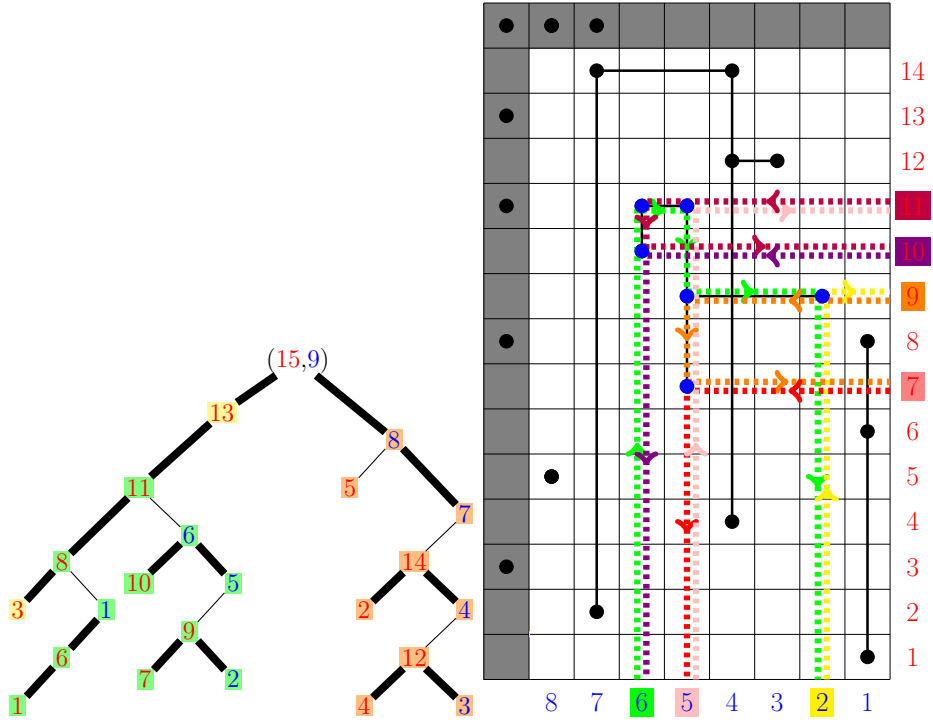

    There is one more hook in the initial tree than in the 3-tuple of sets
    described above as the hook of the root does not appear in this decomposition:
    this first remark justifies the first term $z$ in the product.
    
    The yellow set is constituted by a set of red labels and a set of blue labels, 
    all red (resp. blue) elements contributing to $\LO(T)$ (resp. $\RO(T)$). The 
    generating series associated to this set is then $e^{\alpha x}e^{\beta y}$.
    
    The green set is the set of sub-non-ambiguous trees attached to a 
    vertex of the hook of the root. Each subtree gives a contribution
    $\alpha$ as its root belongs to the hook of the root of the initial
    tree. Following Lemma \ref{lem:ana_enumeration_nb_bloc_fixe}, the 
    associated generating series is then  
    $e^{-\alpha\ln\left(1-z(e^x-1)(e^y-1)\right)}$.
    
    In the same way, the generating series associated with the orange set is         $e^{-\beta\ln\left(1-z(e^x-1)(e^y-1)\right)}$.
    
    For the enumeration formula, let us consider a non-ambiguous tree $T$. We
    use the same idea of decomposition. We delete the topmost row and the
    leftmost column, before using the zigzags paths. Then, for example, one of the zigzag on the tree of Figure~\ref{fig:ana_zigzag_cycle} is illustrated on Figure~\ref{fig:fig_nat2}. For this tree, we get
    \[
        \r{(}\r{3}\r{)}\,
        \r{(}\r{13}\r{)}\,
        \r{(}\b{6}\;\b{2}\;\r{9}\;\r{7}\;\b{5}\;\r{11}\;\r{10}\r{)}\,
        \r{(}\b{1}\;\r{8}\;\r{6}\;\r{1}\r{)}\,
        \b{(}\b{8}\;\r{5}\b{)}\,
        \b{(}\b{7}\;\b{3}\;\r{12}\;\r{4}\;\b{4}\;\r{14}\;\r{2}\b{)}.
    \]
    
    Red (resp. blue) parentheses means that the corresponding sub-non-ambiguous
    tree is attached to the leftmost (resp. rightmost) branch. We 
    illustrate on the right of Figure~\ref{fig:fig_nat2} the cycle 
    $    \r{(}\b{6}\;\b{2}\;\r{9}\;\r{7}\;\b{5}\;\r{11}\;\r{10}\r{)}$.
    This way, we obtain a partition of red (resp. blue)
    labels in $p=hook(T)$ non-empty sets. The number of non-root points in the
    leftmost column (resp. topmost row), with no right (resp. left) child, is
    equal to the number of elements minus 1 in the subset containing the biggest
    element. This explains the $\alpha$-analogue (resp. $\beta$-analogue) of
    $\operatorname S_2$. Let us order and number the $p-1$ other blue subsets
    with respect to their biggest element and pair each blue subset with the
    red block to its right (in the same cycle). Keeping the same example,
    we obtain
    \[
        \begin{array}{ccccccc}
               & 6    & 5  & 4     & 3 & 2 & 1\\
            \{\b{9}\} &\{\b{8}\} &\{\b{7}\;\b{3}\} &\{\b{6}\;\b{2}\}
            &\{\b{5}\} &\{\b{4}\} &\{\b{1}\}\\
            \{\r{15}\;\r{13}\;\r{3}\} &\{\r{5}\} &\{\r{12}\;\r{4}\}
            &\{\r{9}\;\r{7}\} &\{\r{11}\;\r{10}\} &\{\r{14}\;\r{2}\}
            &\{\r{8}\;\r{6}\;\r{1}\}\\
        \end{array}.
    \]

    In a general setting, there are $(p-1)!$ pairing possibilities. Let us now
    replace each pair with its corresponding number. We get
    \[
        \b{(}6\b{)}\;
        \b{(}5\;2\b{)}\;
        \r{(}4\;3\r{)}\;
        \r{(}1\r{)}.
    \]

    In the end, in addition to the two partitions, we have a permutation of size
    $p-1$ decomposed in cycles and whose cycles are coloured in red or in blue.
    Each red (resp. blue) cycle counts for an $\alpha$ (resp. $\beta$), hence,
    the generating polynomial of such permutations is $(\alpha+\beta)_{p-1}$.
    We finally get the desired formula.
\end{proof}

As stated in the introduction, Proposition~\ref{gen_ser_nat} and
Theorem~\ref{thm_p} (in the case $\alpha=\beta=1$) were already proven by Clark
and Ehrenborg \cite{CE}. In the proof they gave, the statistic $p$ is
interpreted on permutations as follows.
\begin{defi}\label{def:ana_CE}
    Let $i,j$ and $n$ be positive integers such that $n=i+j-1$. Let $\p$ be
    a permutation of size $n$ such that all its excedances are at position
    $\llbracket 1,j-1 \rrbracket$. The \emph{CE-statistic} 
    \footnote{
            The "+1" does not appear in the definition of Clark and Ehrenborg. We
            introduced it because there is a shift between the hook statistic and the
            CE-statistic.
        }
         of $\p$ is the
    positive integer
    \[
        \operatorname{CE}(\p)=|\{u\in\llbracket 1,j-1 \rrbracket,\, \p(u)>i\}|+1.
    \]
\end{defi}
    According to Lemma~\ref{lem:ana_enumeration_nb_bloc_fixe}, the hook
    statistic corresponds to the number of blue blocks in 2-coloured block
    decreasing cycles. The bijection between 2-coloured block decreasing cycles
    and permutations with all their excedances at the beginning that we will
    study is $\Theta=\varphi\circ\psi^{-1}$. The following lemma describes the
    difference between the number of blue blocks and the CE-statistic.
    \begin{lem}\label{lem:ana_blocs_bleus_CE}
        Let $i$ and $j$ be positive integers and $c$ a 2-coloured block
        decreasing cycle of size $i\times j$. Then, the number of blue blocks of
        $c$ is equal to the CE-statistic of $\Theta(c)$
        \begin{itemize}
            \item minus 1, if $\b{j}$ has a blue element to its right and
            \item plus 1, if $\r{1}$ has a blue integer, different from $\b{j}$
                to its left.
        \end{itemize}
    \end{lem}
    \begin{proof}
        Let $\p$ be a permutation of size $n=i+j-1$ such that all its excedances
        are at position $\llbracket 1,j-1 \rrbracket$. If we keep the
        interpretation with red and blue integers, the CE-statistic can be interpreted as
        \[
            |\{
                \b{u}\in\llbracket \b{1},\b{j-1} \rrbracket,\,
                \p(\b{u})\in\llbracket \r{2},\r{i} \rrbracket
            \}|+1.
        \]

        Hence, the CE-statistic is the number of blue blocks of $\p$ without
        $\r{1}$ to their right plus 1. Hence, we should study how the number of
        blue blocks behave with respect to $\Theta$. Using
        Lemma~\ref{lem:ana_perm_cycles_cycle_bleus_rouges} we obtain the
        conditions of Lemma~\ref{lem:ana_blocs_bleus_CE}.
    \end{proof}
    Following the previous lemma,  $\Theta$ is not sufficient to prove the
    equidistribution between the hook statistic and the CE-statistic. We need
    one last involution. Let $m$ be the representative of a 2-coloured block
    decreasing sequence $c$ of size $i\times j$ such that $\b{j}$ is on the left of
    $m$. The word $m$ can be factorised as $m=\b{j}b_1\cdots b_k\r{1}m'$ where
    the $b_i$ are maximal blocks of same colours. Let $\omega$ be the involution
    such that if $k$ is even then a representative of $\omega(c)$ is 
    $m=\b{j}b_2b_1\cdots b_kb_{k-1}\r{1}m'$, and if $k$ is odd then
    $\omega(c)=c$.

    \begin{prop}\label{prop:ana_blocs_bleus_CE_equi}
        Let $c$ be a 2-coloured block decreasing sequence. The number of blue
        blocks of $c$ is equal to the CE-statistic of $\Theta(\omega(c))$.
    \end{prop}
    \begin{cor}
        The hook statistic on non-ambiguous trees and the CE-statistic on
        permutations with all their excedances at the beginning are
        equidistributed.
    \end{cor}

\subsection{q-analogues of the hook formula}\label{q_analogs}
As for binary trees, there exist  $q$-analogues of the hook formula for NATs 
of a given shape, associated to either the number of inversions
or the major index. There are two ingredients: first we need to associate a pair of
permutations to a non-ambiguous tree, and second we need to give a $q$-analogue
of the bilinear map $\Bxy$. It turns out that it is possible to use two
different $q$'s namely $q_R$ and $q_L$ for the derivative and integral in $x$
and $y$.

The first step, in order to formulate a $q$-hook formula, is to associate to any
non-empty non-ambiguous  tree $T$ a pair of permutations 
\begin{equation*}
\sigma(T)=(\sigma_L(T),
\sigma_R(T))\in\SG_{\LV(T)} \times \SG_{\RV(T)}.
\end{equation*}

\begin{defi}
  Let $T$ be a non-ambiguous tree. Then $\sigma_L(T)$ is obtained by
  performing a left postfix reading of the left labels: precisely we
  recursively read trees $\NodeBT{L}{R}$ by reading the left labels of $L$,
  then the left labels of $R$ and finally the label of the root if it is a left
  child. The permutation $\sigma_R(T)$ is defined similarly by reading the
  right labels in the right subtree, then in the left subtree and finally reading
  the root.
\end{defi}
If we consider the example of Figure~\ref{fig_exple_nat} the two associated permutations are  $\sigma_L(T) = (2,$ $1,$ $4,$ $3,$ $6,$ $10,$ $8,$ $9,$ $5,$ $7)$ and $\sigma_R(T) = (1,$ $2, 3,$ $4, 5,$ $7, 11,$ $9, 6,$ $8, 10)$.

Recall that the \emph{number of inversions} of a permutation $\sigma\in\SG_n$
is the number of $i<j\leqslant n$ such that $\sigma(i)>\sigma(j)$ (denoted by Inv from now on). A \emph{descent} of
$\sigma$ is a $i<n$ such that $\sigma(i)>\sigma(i+1)$ and the \emph{inverse
  major index} of $\sigma$ is the sum of the descents of $\sigma^{-1}$ (denoted by iMaj from now on). Finally
for a repetition free word $w$ of length~$l$ we denote by $\std(w)$ the
permutations in $\SG_l$ obtained by renumbering $w$ keeping the order of the
letters. For example $\std(36482)=24351$. We define as usual the $q$-integer
$[n]_q:= \frac{1-q^n}{1-q}$, and the $q$-factorial $[n]_q!:= \prod_{i=1}^{n}
[i]_q$. 
For a non-ambiguous tree $\N$ and a statistic $S\in\{\inv,\imaj\}$, we define the weight of a NAT $T$ as
\begin{equation}
  \omega_S(T):= q_L^{S(\sigma_L(T))}q_R^{S(\sigma_R(T))}.
\end{equation}
The following theorem is a $q$-analogue of Proposition \ref{prop_hook_nat}.
\begin{thm}\label{thm-q-hook}

For any non-empty binary tree $B$, 
\begin{equation}
   \sum_{T\in\NAT(\B)} \omega_{S}(T) = 
\frac{ |\LV(\B)|_{q_L}! \cdot |\RV(\B)|_{q_R}!}{
  \displaystyle\prod_{U: \text{left child}}[\EL(U)]_{q_L} \cdot 
  \prod_{U: \text{right child}}[\ER(U)]_{q_R}
}\,.
\end{equation}
\end{thm}
Going back to the non-ambiguous tree of Figure~\ref{fig_exple_nat}, the
inversions numbers are $\inv(\sigma_L(T)) = 11$ and $\inv(\sigma_R(T)) = 7$, so
that $w_{\inv}(T) = q_L^{11}q_R^{7}$. For the inverse major index, we get the permutations $\sigma_L(T)^{-1} = (2, 1, 4, 3, 9, 5, 10, 7, 8, 6)$ and $
\sigma_R(T)^{-1} = (1, 2, 3, 4, 5, 9, 6, 10, 8, 11, 7)$.
Consequently, $\imaj(\sigma_L(T)) = 1 + 3 + 5 + 7 + 9 = 25$
and     $\imaj(\sigma_R(T)) = 6 + 8 + 10 = 24$ so that 
$w_{\imaj}(T) = q_L^{25}q_R^{24}$.

The argument of the proof follows the same path as for the hook formula, using
pumping functions. Recall that the $q$-derivative and $q$-integral are
defined as $\partial_{x, q} x^n:= [n]_qx^{n-1}$ and
$\int_{x, q} u^n\d u:= \frac{x^{n+1}}{[n+1]_q}$.
Then the $(q_L, q_R)$-analogue of the pumping function is given by
\begin{equation}
  \Bxy_q(f, g) = \int_{x,q_L}\int_{y,q_R}
  \partial_{x,q_L}g(u,v)\cdot
  \partial_{y,q_R}f(u,v)
  \diff u\;\diff v.
\end{equation}

We also define recursively $\Bx_q(B)$ by 
$\Bx_{q}\left(\NodeBT{L}{R}\right) = \Bxy_q\left(\Bx_q(L), \Bx_q(R)\right)\,$, with initial conditions $\Bx_q(\emptyset_R):= x$ and $\Bx_q(\emptyset_L):=y$. Then
the main idea is to go through a bilinear function on  permutations. We
write $\QQ\SG$ the vector space of formal sums of permutations. For any
permutation $\sigma\in\SG_n$ we write $\int\sigma=\sigma[n+1]$ the
permutation in $\SG_{n+1}$ obtained by adding $n+1$ at the end. Again we
extend $\int$ by linearity.
\begin{defi}
  The bilinear map 
  $\BSG:\QQ\SG\times\QQ\SG\mapsto\QQ\SG$ is defined
  for $\sigma\in\SG_m$ and $\mu\in\SG_n$ by 
  $$ \BSG(\sigma,\mu) =
    \sum_{\substack{uv\in\SG_{m+n+1} \\ \std(u)=\int\sigma \\ \std(v)=\mu}} uv\,.$$
\end{defi}
For example
$\BSG(21,12)=21345+21435+21534+31425+31524+
41523+32415+32514+42513+43512$.
\begin{lem}
For two non-empty non-ambiguous trees $C$ and $D$, we have
\begin{align*}
  \sum_{T\in\BNAT(C, D)} \sigma_L(T) &=
  \BSG(\sigma_L(C),\sigma_L(D)) \\
  \text{and} \qquad
  \sum_{T\in\BNAT(C, D)} \sigma_R(T) &=
  \BSG(\sigma_R(D),\sigma_R(C)).
\end{align*}
\end{lem}
\begin{proof}
Let $T$ be a NAT appearing in $\BNAT(C, D)$, and let
its shape be the binary tree $B$.
We may write $\sigma_L(T)=uv$ where $u$ is the word given
by the left postfix reading of the left labels of the left 
subtree of $T$, and $v$ the one of its right subtree.
By definition, we have $\std(u)=\int\sigma_L(C)$ and $\std(v)=\sigma_L(D)$. The symbol $\int$ in the first equality comes from the fact that the root of $C$ becomes a left
label in $T$.
Moreover, the choice of a $v$ such that $\std(v)=\sigma_L(D)$
is equivalent to the choice of a labelling of the left vertices
of the right subtree $B_R$ of $B$, such that the order 
of the labels in $D$ is respected. The same holds for $u$,
and the corresponding root $C$.
This proves the first equality.

The second one is obtained in the same way.
We have to take care of two things: 
that the right postfix reading
reads first the right subtree, 
and that the root of $D$ here becomes a right label in $T$.
The order of $C$ and $D$ is thus reversed in $\BSG$. 
\end{proof}

Let us define a bilinear map $\Phi_S$. For two permutations 
$\tau\in\SG_{m}$ and $\pi\in\SG_{n}$
\begin{equation}
\Phi_S((\tau,\pi)):=
q_L^{S(\tau)}\frac{x^{m+1}}{[m+1]_{q_L}!}\,
q_R^{S(\pi)} \frac{y^{n+1}}{[n+1]_{q_R}!}\,
\end{equation}

We shall use the same notation
for NATs: $\Phi_S(T)= \Phi_S(\sigma_L(T),\sigma_R(T))$.

The main result, analogous to Proposition \ref{prop:commutation_phi} is the following.
\begin{prop}\label{prop:qhook-commut}
 For any non-empty binary tree $B$, we have the relation
$$\Phi_S(\NAT(B)) = \Bx_q(B).$$
\end{prop}
As for Proposition \ref{prop:commutation_phi}, 
Proposition \ref{prop:qhook-commut} is derived
recursively.
The analogue of Lemma \ref{lem:com_phi} 
is the following.

\begin{lem}\label{lem:com_phi_q}
For $(T_1, T_2)$ a pair of NATs in $\QQ\NAT_L\times\QQ\NAT_R$, one has
  $$\Phi_S(\BNAT(T_1, T_2)) = \Bxy_q(\Phi_S(T_1),\Phi_S(T_2)).$$
\end{lem}

To prove it, we first need a technical result.
\begin{lem}\label{lem:technical-q}
For a statistic $S\in\{\inv,\imaj\}$,
and for $\tau\in\SG_{m}$ and $\pi\in\SG_{n}$, we have:
$$
\sum_{\substack{\theta=uv\in\SG_{n+m+1} \\ \std(u)=\int\tau \\ \std(v)=\pi}}
q^{S(\theta)}
=
q^{S(\tau)+S(\pi)}\binom{m+n+1}{m+1}_{q}.
$$
\end{lem}
\begin{proof}
The case $S=\inv$ is easier to prove. The $q$-binomial consists in choosing a
permutation $\theta$ such that $\theta(1)<\cdots<\theta(m+1)$ and
$\theta(m+2)<\cdots<\theta(m+n+1)$. The term $q^{S(\tau)+S(\pi)}$ comes from the reordering of the $\theta(i)$ in order to have $\std(\theta(1)\cdots\theta(m+1))=\int\tau$ and $\std(\theta(m+2)\cdots\theta(m+n+1))=\pi$.

In order to prove the case $S=\imaj$, we first obtain the following two equations
\begin{equation}\label{eq:imaj_lemma_1}
\sum_{\substack{\theta=uv\in\SG_{n+m+1} \\ n+m+1\in u \\ \std(u)=\int\tau \\ \std(v)=\pi}}
q^{\imaj(\theta)}
=
q^{\imaj(\tau)+\imaj(\pi)+n}\binom{m+n}{n}_{q},
\end{equation}
and
\begin{equation}\label{eq:imaj_lemma_2}
\sum_{\substack{\theta=uv\in\SG_{n+m+1} \\ n+m+1\in v \\ \std(u)=\int\tau \\ \std(v)=\pi}}
q^{\imaj(\theta)}
=
q^{\imaj(\tau)+\imaj(\pi)}\binom{m+n}{n-1}_{q}.
\end{equation}

Equation~\eqref{eq:imaj_lemma_1} is a consequence of Equation~36 in
\cite{HivNovThib08}. We prove Equation~\eqref{eq:imaj_lemma_2} by induction on
$n+m$ and by distinguishing the two cases $n+m\in u$ and $n+m\in v$.
In the first case, by using \eqref{eq:imaj_lemma_1} we have
$$
\sum_{\substack{\theta=uv\in\SG_{n+m+1} \\ n+m+1\in v \\
n+m\in u\\ \std(u)=\int\tau \\ \std(v)=\pi}}
q^{\imaj(\theta)}
=
\sum_{\substack{uv'\in\SG_{n+m+1} \\ 
n+m\in u\\ \std(u)=\int\tau \\ \std(v')=\pi'}}
q^{\imaj(uv')}
=
q^{\imaj(\tau)+\imaj(\pi')+n-1}\binom{m+n-1}{n-1}_{q},
$$
where $\pi'$ is obtained from $\pi$ by deleting the entry $n$.
Let us set $\epsilon$ equal to $1$ if $n$ is to the left of $n-1$ 
in $\pi$ and $0$ otherwise. Then we may write: 
$\imaj(\pi)=\imaj(\pi')+\epsilon(n-1)$, whence
$$
\sum_{\substack{\theta=uv\in\SG_{n+m+1} \\ n+m+1\in v \\
n+m\in u\\ \std(u)=\int\tau \\ \std(v)=\pi}}
q^{\imaj(\theta)}
=
q^{\imaj(\tau)+\imaj(\pi')+(1-\epsilon)(n-1)}\binom{m+n-1}{n-1}_{q}.
$$

In the second case, by using \eqref{eq:imaj_lemma_1} by induction, we get:
\begin{align*}
\sum_{\substack{\theta=uv\in\SG_{n+m+1} \\ n+m+1\in v \\
n+m\in v\\ \std(u)=\int\tau \\ \std(v)=\pi}}
q^{\imaj(\theta)}
&=q^{\epsilon(m+n)}\sum_{\substack{uv'\in\SG_{n+m+1} \\
n+m\in v'\\ \std(u)=\int\tau \\ \std(v')=\pi'}}
q^{\imaj(uv')}\\
&=q^{\imaj(\tau)+\imaj(\pi')+\epsilon(m+n)}\binom{m+n-1}{n-2}_{q}\\
&=q^{\imaj(\tau)+\imaj(\pi)+\epsilon(m+1)}\binom{m+n-1}{n-2}_{q}.
\end{align*}

Then for any of $\epsilon$, we have
$$
q^{\epsilon(m+1)}\binom{m+n-1}{n-2}_{q}
+
q^{(1-\epsilon)(m+1)}\binom{n-1}{n-1}_{q}
=
\binom{m+n}{n-1}_{q}.
$$

Since \eqref{eq:imaj_lemma_1} is trivially true when $\pi$ is the unique permutation of size $1$, the proof by induction is complete.

To conclude the proof of the lemma, we add Equations~\eqref{eq:imaj_lemma_1} and \eqref{eq:imaj_lemma_2}
and use the identity
$$
\binom{a+b+1}{a}_{q}=
\binom{a+b}{a}_{q}
+
\binom{a+b}{a-1}_{q}.$$
\end{proof}

\begin{proof}[Proof of Lemma \ref{lem:com_phi_q}]
Let us set 
$$
u=\Phi_S(T_1)=
q_L^{S(\sigma_L(T_1))}\frac{x^{w_L(T_1)}}{[w_L(T_1)]_{q_L}!}\,
q_R^{S(\sigma_R(T_1))} \frac{y^{w_R(T_1)}}{[w_R(T_1)]_{q_R}!}\,
$$
and
$$
v=\Phi_S(T_2)=
q_L^{S(\sigma_L(T_2))}\frac{x^{w_L(T_2)}}{[w_L(T_2)]_{q_L}!}\,
q_R^{S(\sigma_R(T_1))} \frac{y^{w_R(T_1)}}{[w_R(T_1)]_{q_R}!}\,.
$$

Thanks to Lemma \ref{lem:technical-q}, we have to prove that
\begin{align*}
\Bxy_q(u,v)
&=
q_L^{S(\sigma_L(T_1))+S(\sigma_L(T_2))}
\binom{w_L(T_1)+w_L(T_2)-1}{w_L(T_1)}_{q_L}\\
&\times
q_R^{S(\sigma_R(T_1))+S(\sigma_R(T_2))}
\binom{w_R(T_1)+w_R(T_2)-1}{w_R(T_2)}_{q_R}\\
&\times
\frac
{x^{w_L(T_1)+w_L(T_2)}\,y^{w_R(T_1)+w_R(T_2)}}
{[w_L(T_1)+w_L(T_2)]_{q_L}!\,[w_R(T_1)+w_R(T_2)]_{q_R}!}
.
\end{align*}

This is done by a simple computation, as 
in the proof of Lemma \ref{lem:com_phi}.
\end{proof}

\begin{proof}[Proof of Proposition \ref{prop:qhook-commut}]
The proof is exactly the same as for Proposition \ref{prop:commutation_phi}. We just have to replace $\Phi$
by $\Phi_S$ and $\Bx$ by $\Bx_q$.
\end{proof}

We are now in a position to prove Theorem \ref{thm-q-hook}.
\begin{proof}[Proof of Theorem \ref{thm-q-hook}]
By definition, we have for any NAT $T$: $\Phi_S(T)=\omega_S(T)\times
\frac{x^{w_L(T)}}{[w_L(T)]_{q_L}!}\frac{y^{w_R(T)}}{[w_R(T)]_{q_L}!}$. 
Proposition \ref{prop:qhook-commut} implies that 
for a binary tree $B$, $\Phi_S([B])=\Bx_q(B)$, and
we may compute $\Bx_q(B)$ recursively.
Let us suppose that neither $B_L$ nor $B_R$ is empty.
We then have:
\begin{align*}
\Bx_q(B)
&=\Bxy_q(\Bx_q(B_L),\Bx_q(B_R))\\ 
&=\int_{x,q_L}\int_{y,q_R}
  \partial_{x,q_L}\Bx_q(B_L)(u,v)\cdot
  \partial_{y,q_R}\Bx_q(B_R)(u,v)
  \diff u\;\diff v\\
&=\frac{[w_R(B_L)]_{q_L}[w_L(B_R)]_{q_R}}{[w_L(B)]_{q_L}[w_R(B)]_{q_R}}\Bx_q(B_L)\Bx_q(B_R).
\end{align*}

The same recursive computation as in the proof of Proposition \ref{prop_hook_nat} leads to 
$$
\frac{\sum_{T\in\NAT(B)}\omega_S(T)}{[w_L(B)]_{q_L}![w_R(B)]_{q_R}!}
=
\frac{1}{
  \displaystyle\prod_{U: \text{left child}}[\EL(U)]_{q_L} \cdot 
  \prod_{U: \text{right child}}[\ER(U)]_{q_R}
}
$$
which gives Theorem \ref{thm-q-hook}.
\end{proof}

We conclude this section by an example, using the same tree as in Example \ref{exple_hook}.

\begin{exple}
  
Let $B = \scalebox{0.5}
{ \newcommand{\nodea}{\node[draw,fill,circle] (a) {$$};}
  \newcommand{\nodeb}{\node[draw,fill,circle] (b) {$$};}
  \newcommand{\nodec}{\node[draw,fill,circle] (c) {$$};}
  \newcommand{\noded}{\node[draw,fill,circle] (d) {$$};}
  \newcommand{\nodee}{\node[draw,fill,circle] (e) {$$};}
  \newcommand{\nodef}{\node[draw,fill,circle] (f) {$$};}
  \newcommand{\nodeg}{\node[draw,fill,circle] (g) {$$};}
  \newcommand{\nodeh}{\node[draw,fill,circle] (h) {$$};}
\begin{tikzpicture}[scale=0.5,baseline=(current bounding box.center),
                      every node/.style={inner sep=2pt}]
\matrix[column sep=.15cm, row sep=.15cm,ampersand replacement=\&]{
         \&         \&         \& \nodea  \&         \&         \&         \&         \&         \\
         \& \nodeb  \&         \&         \&         \&         \&         \& \nodee  \&         \\
 \nodec  \&         \& \noded  \&         \&         \& \nodef  \&         \&         \& \nodeh  \\
         \&         \&         \&         \&         \&         \& \nodeg  \&         \&         \\
};
\path[thick] (b) edge (c) edge (d)
    (f) edge (g)
    (e) edge (f)
    (a) edge (b) edge (e)
        (e) edge (h);
\end{tikzpicture}}$.
The $q$- hook formula is given by:
\begin{align*}
   \sum_{T\in\NAT(\B)} \omega_{S}(T) &= 
\frac{ [3]_{q_L}! [4]_{q_R}!}{
  \left([1]_{q_L}[2]_{q_L}[1]_{q_L} \right) \cdot
  \left( [1]_{q_R}[1]_{q_R}[3]_{q_R}[1]_{q_R}\right)
} \\
&=(q_R^3 + q_R^2 + q_R + 1)(q_L^2 +
q_L + 1)(q_R + 1) \,.
\end{align*}
 Expanding this expression, one finds that the coefficient of
$q_R^2q_L$ is $2$.

For the $\imaj$ statistic it corresponds to the two following non-ambiguous
trees which are shown with their associated left and right permutations:
\[
{\newcommand{\nodea}{\node (a) {(\color{red}$4$,\color{blue}$5$)}
;}\newcommand{\nodeb}{\node (b) {\color{red}$3$}
;}\newcommand{\nodec}{\node (c) {\color{red}$2$}
;}\newcommand{\noded}{\node (d) {\color{blue}$2$}
;}\newcommand{\nodee}{\node (e) {\color{blue}$4$}
;}\newcommand{\nodef}{\node (f) {\color{red}$1$}
;}\newcommand{\nodeg}{\node (g) {\color{blue}$3$}
;}\newcommand{\nodeh}{\node (h) {\color{blue}$1$}
;}\begin{tikzpicture}[baseline=(current bounding box.center)]
\matrix[column sep=1.5mm, row sep=1.5mm,ampersand replacement=\&]{
         \&         \&         \& \nodea  \&         \&         \&         \&         \&         \\ 
         \& \nodeb  \&         \&         \&         \&         \&         \& \nodee  \&         \\ 
 \nodec  \&         \& \noded  \&         \&         \& \nodef  \&         \&         \& \nodeh  \\ 
         \&         \&         \&         \&         \&         \& \nodeg  \&         \&         \\
};

\path (b) edge (c) edge (d)
    (f) edge (g)
    (e) edge (f) edge (h)
    (a) edge (b) edge (e);
\end{tikzpicture}} \left((2, 3, 1), (1, 3, 4, 2)\right),\]

\[{ \newcommand{\nodea}{\node (a) {(\color{red}$4$,\color{blue}$5$)}
;}\newcommand{\nodeb}{\node (b) {\color{red}$3$}
;}\newcommand{\nodec}{\node (c) {\color{red}$2$}
;}\newcommand{\noded}{\node (d) {\color{blue}$2$}
;}\newcommand{\nodee}{\node (e) {\color{blue}$4$}
;}\newcommand{\nodef}{\node (f) {\color{red}$1$}
;}\newcommand{\nodeg}{\node (g) {\color{blue}$1$}
;}\newcommand{\nodeh}{\node (h) {\color{blue}$3$}
;}\begin{tikzpicture}[baseline=(current bounding box.center)]
\matrix[column sep=1.5mm, row sep=1.5mm,ampersand replacement=\&]{
         \&         \&         \& \nodea  \&         \&         \&         \&         \&         \\ 
         \& \nodeb  \&         \&         \&         \&         \&         \& \nodee  \&         \\ 
 \nodec  \&         \& \noded  \&         \&         \& \nodef  \&         \&         \& \nodeh  \\ 
         \&         \&         \&         \&         \&         \& \nodeg  \&         \&         \\
};

\path (b) edge (c) edge (d)
    (f) edge (g)
    (e) edge (f) edge (h)
    (a) edge (b) edge (e);
\end{tikzpicture}}
\left((2, 3, 1), (3, 1, 4, 2)\right).\]

For the $\inv$ statistic it corresponds to the two following non-ambiguous
trees which are shown with their associated left and right permutations:
\[
{\newcommand{\nodea}{\node (a) {(\color{red}$4$,\color{blue}$5$)}
;}\newcommand{\nodeb}{\node (b) {\color{red}$3$}
;}\newcommand{\nodec}{\node (c) {\color{red}$1$}
;}\newcommand{\noded}{\node (d) {\color{blue}$3$}
;}\newcommand{\nodee}{\node (e) {\color{blue}$4$}
;}\newcommand{\nodef}{\node (f) {\color{red}$2$}
;}\newcommand{\nodeg}{\node (g) {\color{blue}$1$}
;}\newcommand{\nodeh}{\node (h) {\color{blue}$2$}
;}\begin{tikzpicture}[baseline=(current bounding box.center)]
\matrix[column sep=1.5mm, row sep=1.5mm,ampersand replacement=\&]{
         \&         \&         \& \nodea  \&         \&         \&         \&         \&         \\ 
         \& \nodeb  \&         \&         \&         \&         \&         \& \nodee  \&         \\ 
 \nodec  \&         \& \noded  \&         \&         \& \nodef  \&         \&         \& \nodeh  \\ 
         \&         \&         \&         \&         \&         \& \nodeg  \&         \&         \\
};

\path (b) edge (c) edge (d)
    (f) edge (g)
    (e) edge (f) edge (h)
    (a) edge (b) edge (e);
\end{tikzpicture}} \left((1,3,2), (2,1,4,3)\right),\]

\[{ \newcommand{\nodea}{\node (a) {(\color{red}$4$,\color{blue}$5$)}
;}\newcommand{\nodeb}{\node (b) {\color{red}$3$}
;}\newcommand{\nodec}{\node (c) {\color{red}$1$}
;}\newcommand{\noded}{\node (d) {\color{blue}$2$}
;}\newcommand{\nodee}{\node (e) {\color{blue}$4$}
;}\newcommand{\nodef}{\node (f) {\color{red}$2$}
;}\newcommand{\nodeg}{\node (g) {\color{blue}$3$}
;}\newcommand{\nodeh}{\node (h) {\color{blue}$1$}
;}\begin{tikzpicture}[baseline=(current bounding box.center)]
\matrix[column sep=1.5mm, row sep=1.5mm,ampersand replacement=\&]{
         \&         \&         \& \nodea  \&         \&         \&         \&         \&         \\ 
         \& \nodeb  \&         \&         \&         \&         \&         \& \nodee  \&         \\ 
 \nodec  \&         \& \noded  \&         \&         \& \nodef  \&         \&         \& \nodeh  \\ 
         \&         \&         \&         \&         \&         \& \nodeg  \&         \&         \\
};

\path (b) edge (c) edge (d)
    (f) edge (g)
    (e) edge (f) edge (h)
    (a) edge (b) edge (e);
\end{tikzpicture}}
\left((1,3,2), ( 1,3, 4, 2)\right).\]

\end{exple}

\section{Non-ambiguous trees in higher dimension}
\label{gnat}

In this section we give a generalisation of NATs to higher dimensions.
NATs are defined as binary trees whose vertices are 
embedded in $\mathbb{Z}^2$, and edges are objects of dimension 1 (segments).
Let $d \ge k \ge 1$ be two integers.
In higher dimension, binary trees are replaced by $\binom{d}{k}$-ary trees 
embedded in $\mathbb{Z}^d$ 
and edges are objects of dimension $k$.
As in Section~\ref{differential} we obtain differential equations for these objects.

\subsection{Definitions}
\label{def_natdk}

We call \emph{$(d,k)$-direction} a subset of cardinality $k$ of $\{1, \ldots, d\}$.
The set of $(d,k)$-directions is denoted by $\edir$. 
A $(d,k)$-tuple is a $d$-tuple of $(\mathbb{N}\cup\{\bullet\})^d$, 
in which $k$ entries are integers and $d-k$  are $\bullet$.
For instance,
$
(\bullet, \ 1, \  \bullet, \  5,\  2, \ \bullet, \ \bullet, \ 3, \ \bullet)
$
is a $(9,4)$-tuple.
The direction of a $(d,k)$-tuple $U$ is the set of indices of $U$
corresponding to entries different from $\bullet$.
For instance, the direction of our preceding example is $\{2,4,5,8\}$.

\begin{defi} 
A \emph{$\binom{\d}{\k}$-ary} tree $\M$ is a tree in which the children of a given vertex
are indexed by a $(\d,\k)$-direction.
\end{defi}

A $\binom{\d}{\k}$-ary tree will be represented as an ordered tree where the 
children of a vertex $S$ are drawn from left to right with respect to the 
lexicographic order of their indices.
If a vertex $S$ has no child associated to an index $\dir$, we draw an half 
edge in this direction.
Two examples are drawn in Figure~\ref{fig_gnat}.
As for binary trees, for each $(d,k)$-direction $\dir$ we consider that there is
a $\binom{\d}{\k}$-ary tree of size 0: the \emph{empty $\binom{\d}{\k}$-ary
tree of direction $\dir$} noted $\emptyset_\dir$.

\begin{defi}\label{def_gnat}
A \emph{non-ambiguous tree of dimension $(d,k)$} is a labelled 
\emph{$\binom{d}{k}$-ary} tree such that:
\begin{enumerate}[ref={\thecor.\arabic*}]
\item\label{gnat_dk_tuple} a child of index $\dir$ is labelled with a
    $(d,k)$-tuple of direction $\dir$ and \label{gnat_root} the root is
    labelled with a $(d,d)$-tuple;
\item\label{gnat_growth} for any descendant $U$  of $V$, if the $i$-th
    component of $U$ and $V$ are different from $\bullet$, then the $i$-th
    component of $V$ is strictly greater than the $i$-th component of $U$;
\item\label{gnat_distinct} for each $i \in \i{1}{\d}$, all the $i$th
    components, different from $\bullet$, are pairwise distinct;
\item\label{gnat_interval} the set of $i$th components different from
    $\bullet$ of every vertices in the tree is an interval whose minimum is
    $1$.
\end{enumerate}

The set of non-ambiguous trees of dimensions $(d,k)$ is denoted by $\NATdk$.
\end{defi}

\begin{rque}
The usual NATs are non-ambiguous trees of dimension  $(2,1)$.
\end{rque}

We write \nat for a non-ambiguous tree (of dimensions $(d,k)$).
Figure~\ref{fig_gnat} gives an example of a \nat[3][1] and a \nat[3][2].

\begin{figure}[h]
    \begin{center}
    \resizebox{\textwidth}{!}{
\includegraphics[scale=0.3]{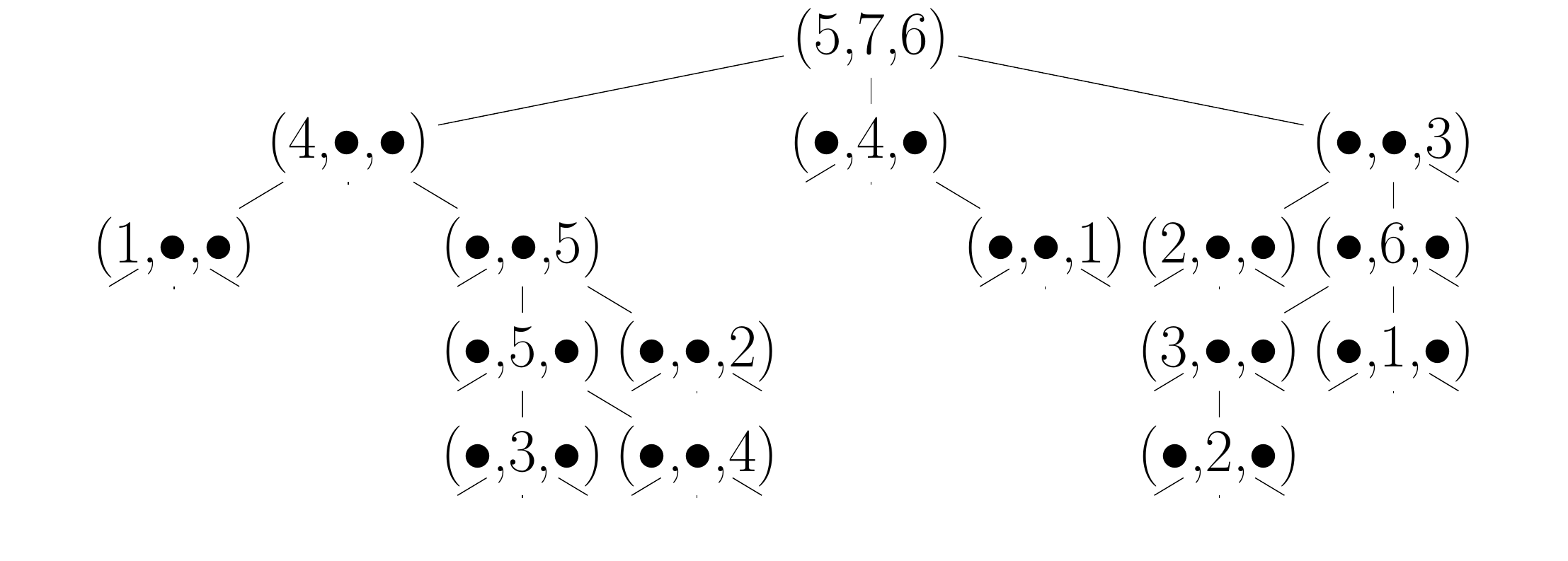}
\includegraphics[scale=0.5]{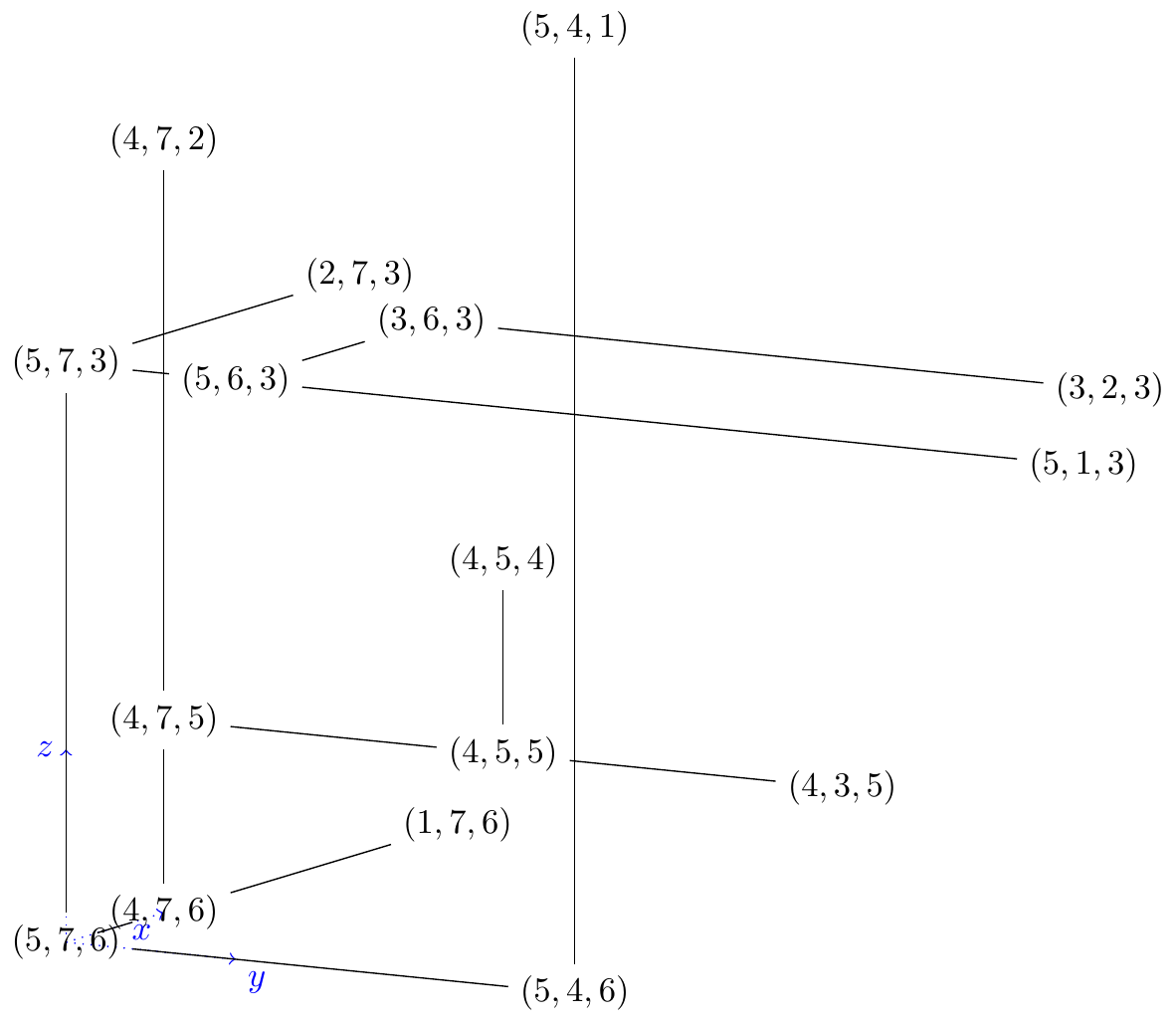}}

        \hspace{1cm}
\includegraphics[scale=0.5]{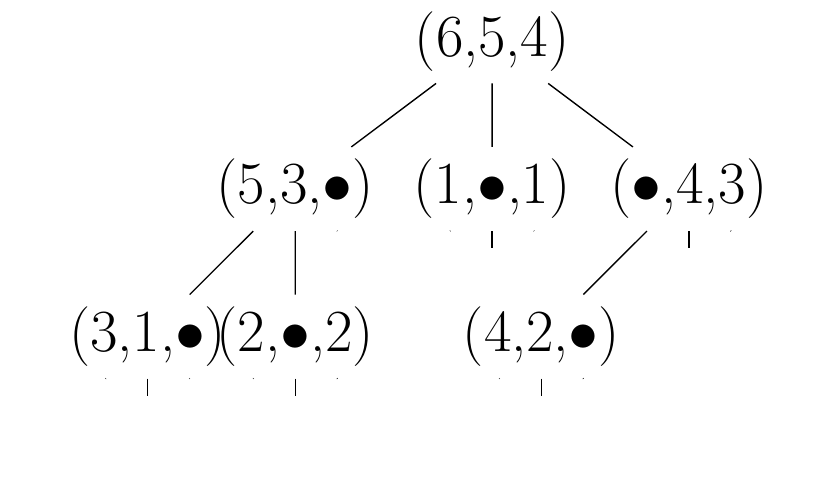}
\includegraphics[scale=0.5]{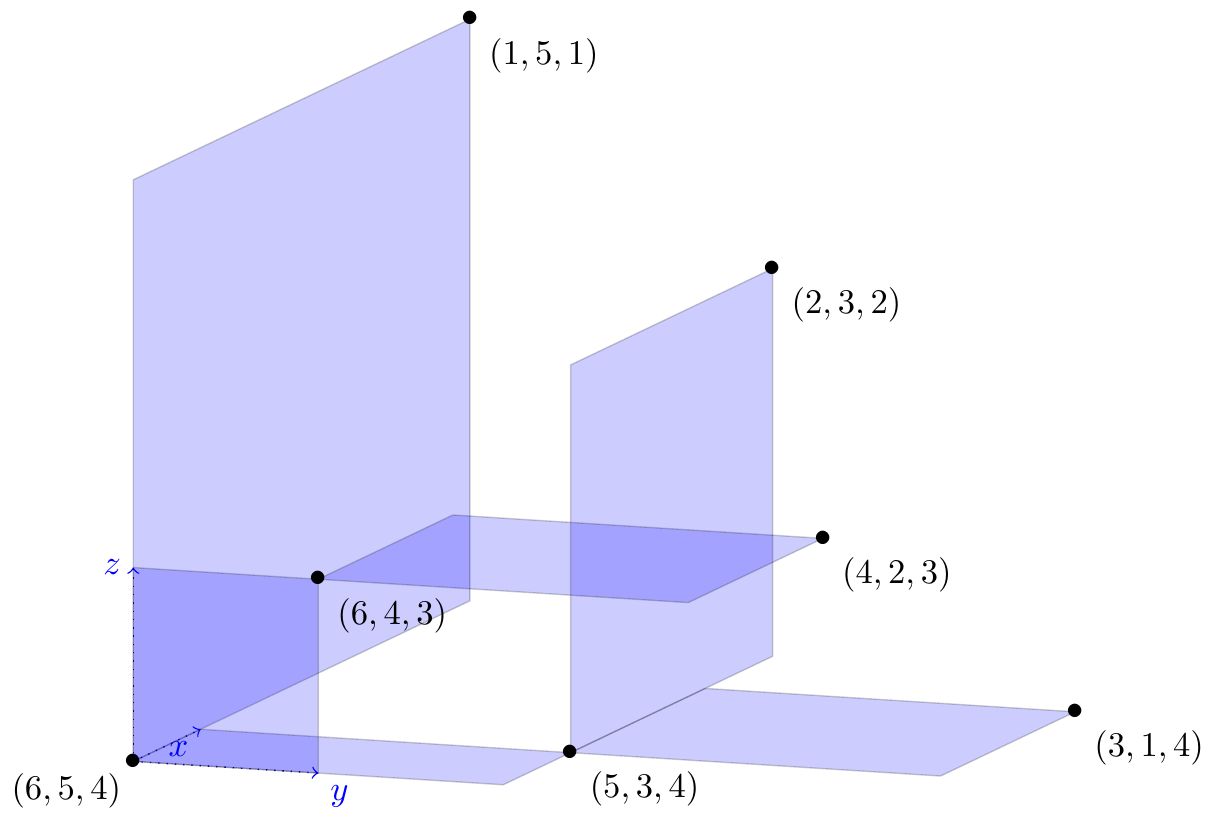}
    \end{center}
    \caption{A NAT of dimension $(3,1)$ and its geometric representation (above) and a NAT of dimension $(3,2)$ and its geometric representation
    (below).} 
    \label{fig_gnat}\label{ex_kdtree}
\end{figure}

\begin{defi}\label{gnat_gs}

The \emph{geometric size} of a \nat is the $d$-tuple of integers 
$(w_1, \ldots, w_d)$ which labels the root of the \nat, it is denoted by $\gs$.
The \emph{$\dir$-size} of a \nat   is the number of vertices in the tree of
direction $\dir$, the set of such vertices is denoted by $\DV$. 
\end{defi}

Proposition~\ref{prop_relation_sizes} gives the relation between the 
geometric size and the $\dir$-size of a non-ambiguous trees.

\begin{prop}
\label{prop_relation_sizes}
Let $\M$ be a $\binom{d}{k}$-ary tree, the root label is constant on elements of
$\NATdk$ of shape $M$ ($\NATdk(\M)$):
$$
\wi_i(M)
:=
\wi_i
=
\sum_{\dir \in \edir \ \mid\ i \in \dir} |\DV(\M)| +1
.
$$
\end{prop}

\subsection{Associated differential equations}
As in section~\ref{differential}, we define a weight $\Phi$, a
$\binom{d}{k}$-linear map $\BNAT$ on generalised non-ambiguous trees of dimension
$(d,k)$ and $\binom{d}{k}$-linear map $\Bxy$ on multivariate power series such
that ``$\Phi\circ\BNAT=\Bxy\circ\Phi$''. From this identity, we obtain a hook formula
for the number of generalised non-ambiguous trees with a fixed underlying tree
and a differential equation satisfied by the generating function of
generalised non-ambiguous trees of dimension $(d,k)$.

Let $\emptyset_\dir$ be the empty $\binom{d}{k}$-ary tree of direction $\dir$, by convention
\[
    \wi_i(\emptyset_\dir)=
    \left \{
        \begin{array}{ll}
            0   &\text{ if } i\in\dir,\\
            1   &\text{ else.}
        \end{array}
    \right.
\]

The weight of a generalised non-ambiguous tree $\N$ is given by
\[
  \Phi(\N) := 
  \prod_{i=1}^d
  \frac{
    x_i^{\wi_i(\N)}
  }{
    \wi_i(\N) !
  }.
\]

Let $\QQ\NATdk$ be the $\QQ$-vector space generated by generalised non-ambiguous
trees of dimension $(d,k)$, and $\QQ\NATdkp$ its subspace generated by non-empty
generalised non-ambiguous trees and $\emptyset_\dir$. Let $\BNAT$ be the
multilinear map
\[
    \BNAT:
    \displaystyle\prod_{\dir\in\edir}\mathbb{Q}\NATdkp
    \rightarrow
    \mathbb{Q}\NATdk
\]
such that 
\(
    \BNAT((NAT_\dir)_{\dir\in\edir})
\)
is equal to the formal sum of all the generalised non-ambiguous trees such that
the root's child of direction $\pi$ is $NAT_\dir$.
To define $\Bxy$, we need the following notations
\[
  \partial_\dir:=\partial_{i_1}\cdots\partial_{i_k}
  ,\hspace{.5cm}
  \int_{\dir}:=
  \int_0^{x_{i_1}}
  \cdots
  \int_0^{x_{i_k}}
  \hspace{.5cm}
  \text{and}
  \hspace{.5cm}
  \d x_\dir:=\d x_{i_1}\cdots\d x_{i_k}.
\]
with $\dir=\{i_1,\ldots,i_k\}$. We denote by $\Bxy$ the $\binom{d}{k}$-linear map
\[
    \QQ\llbracket x_1,\ldots,x_d \rrbracket^{\edir}
    \rightarrow
    \QQ\llbracket x_1,\ldots,x_d \rrbracket
\]
such that
\[
    \Bxy((f_\dir)_{\dir\in\edir}):=
    \mathlarger{\int}_{\llbracket 1,d \rrbracket}
    \left(
        \prod_{\dir\in\edir}
        \partial_{\llbracket 1,d \rrbracket\setminus\dir}(f_\dir)
    \right)
    \;\d x_{\llbracket 1,d \rrbracket}.
\]

Similarly to the dimension $(2,1)$, by recursion on $\binom{d}{k}$-ary trees, we prove that
\begin{equation}\label{eq:commutative_diag}
  \Bxy(\Phi(NAT_\dir)_{\dir\in\edir})
  =
  \Phi(\BNAT((NAT_\dir)_{\dir\in\edir})).
\end{equation}

As for non-ambiguous trees (Proposition \ref{prop_hook_nat}), there is a hook
formula for the number of non-ambiguous trees of dimension $(d,k)$ with fixed
underlying tree. Let $\M$ be a $\binom{d}{k}$-ary tree, for each vertex $U$ we
denote by $\E_i(U)$ the number of vertices, of the subtree whose root is $U$
(itself included in the count), whose direction contains $i$. Then, by
recursively using the previous equation we obtain that
$$
|\NATdk(M)|
=
\frac{
    \displaystyle
    \prod_{i=1}^d \left( \wi_i( \M ) - 1 \right) !
}{
    \displaystyle
    \prod_{
        \substack{
            U: \text{ child of direction}\\ \text{ containing } i
        }
    }{
        \E_i(U)
    }
}
.
$$

Let $\LGFNdk(x_1,\ldots,x_d)$ be the exponential generating function of
non-empty generalised non-ambiguous trees
\[
    \LGFNdk(x_1,\ldots,x_d)
    := 
    \sum_{\N \in \NATdk^*}\Phi(\N).
\]
and $\GFNdk(x_1,\ldots,x_d)$ its derivative
\[
  \GFNdk(x_1,\ldots,x_d)
  := 
  \partial_{\{1,\ldots,d\}}\LGFNdk(x_1,\ldots,x_d).
\]

There is a $(d,k)$-dimensional analogue of the fixed point differential
Equation~\ref{equ_gfn}. Similarly to the notations at the beginning of the
section, we define $x_{\{i_1,\ldots,i_k\}}$ as the product $x_{i_1}\cdots
x_{i_k}$, in particular, $\Phi(\emptyset_\dir) = x_{\llbracket 1,d \rrbracket\setminus\dir}$.
\begin{prop}
\label{prop_equ_diff_natdk}
The exponential generating function $\LGFNdk$ of generalised non-ambiguous trees
satisfies the following differential equation
\begin{equation}
  \label{equ_lgfndk}
  \LGFNdk
  =
  \sum_{\dir\in\edir}x_{\llbracket 1,d \rrbracket\setminus\dir}
  +
  \int_{\llbracket 1,d \rrbracket}
  \prod_{\dir\in\edir}
  \partial_{\llbracket 1,d \rrbracket\setminus\dir}(\LGFNdk)
  \;\d x_{\llbracket 1,d \rrbracket},\\
\end{equation}
and its derivative $\GFNdk$ satisfies
\begin{equation}
\label{equ_gfndk}
\GFNdk
=
\prod_{\dir \in \edir} \left( 1 + \int_{\dir} \GFNdk  \right).
\end{equation}
\end{prop}
\begin{proof}
  The summation of Equation~\ref{eq:commutative_diag} over
  $\displaystyle\prod_{\dir\in\edir}\mathbb{Q}\NATdkp$ gives us
\[
  \LGFNdk
  =
  \sum_{\dir\in\edir}x_{\llbracket 1,d \rrbracket\setminus\dir}
  +
  \int_{\llbracket 1,d \rrbracket}
  \prod_{\dir\in\edir}
  \partial_{\llbracket 1,d \rrbracket\setminus\dir}(\LGFNdk-\sum_{\dir'\in(\edir\setminus\dir)}x_{
    \llbracket 1,d \rrbracket\setminus\dir'})
  \;\d x_\dir.\\
\]
Since $\partial_{\llbracket 1,d \rrbracket\setminus\dir} (x_{\llbracket 1,d
\rrbracket\setminus\dir'})= 0$ for $\dir\neq\dir'$, the
Equation~\ref{equ_lgfndk} is proven. Equation~\ref{equ_gfndk} is obtained by
using the identity
\[
  \GFNdk(x_1,\ldots,x_d)
  = 
  \partial_{\{1,\ldots,d\}}\LGFNdk(x_1,\ldots,x_d).
\]

\end{proof}

In the generic case, we are not able to solve those differential equations.
We know that setting a variable $x_d$ to $0$ gives the generating function 
of NATs of lower dimension.
\begin{prop}
Let $d > k \ge 1$, then
$
\left.
    \GFN_{d,k}
\right|_{x_d=0}
=
\GFN_{d-1,k}
.
$
\end{prop}

For some specific values of $d$ and $k$ we have (at least partial) results.
\begin{prop}
\label{prop_sol_d_d-1}
Let $k=d-1$, if we know a particular solution $s( x_1, \ldots, x_d)$ for
$$
\partial_1 \ldots \partial_d \LGFNdk[d][d-1]
=
\partial_{1} \LGFNdk[d][d-1] \times \ldots \times \partial_{d} \LGFNdk[d][d-1] 
$$
then, for any function $s_1(x_1), \ldots, s_d(x_d)$, the function
$s( s_1(x_1), \ldots, s_d(x_d))$ is also a solution.
\end{prop}

\begin{prop}
\label{prop_sol_d_1}
Some non trivial rational functions are solutions of
$$
\partial_1 \ldots \partial_d \LGFNdk[d][1]
=
\prod_{\dir \in \edir[d][d-1]}
    \partial_{\dir} \LGFNdk[d][1]
.
$$
\end{prop}

\begin{proof}[Proof (sketch)]
We generalise the first part of the proof of
Proposition~\ref{prop:gf_expressions}.
We  define
$
\LGFN_{(i)} = \partial_{\dir} \LGFNdk[d][1]
$
where 
$i \in \i{1}{d}$
and
$\dir = \i{1}{d} \setminus \{i\}$.
We get the relation
$
\partial_i \LGFN_{(i)}
=
\prod_{j=1}^d \LGFN_{(j)}
$
and then
$
\prod_{i=1}^d
\partial_i \LGFN_{(i)}
=
\prod_{i=1}^d \LGFN_{(i)}^d
.
$
To obtain a particular solution, we just need to identify, in the previous 
equation, the term $\partial_i \LGFN_{(i)}$ to the term $\LGFN_{(i)}^d.$
We thus obtain some non trivial solutions for our equation,
which are rational functions.
\end{proof}

Since dimension $(2,1)$ is the unique case where Proposition~\ref{prop_sol_d_d-1}
and Proposition~\ref{prop_sol_d_1} can be applied at the same time,
and the computation of $\GFNdk[\d][\d]$ is straightforward,
we have the following proposition.
\begin{prop}
We have the closed formulas:
$$
\GFNdk[2][1] = \GFN \qquad\text{and}\qquad
\GFNdk[d][d]=\sum_{n \ge 0} \frac{(x_1\cdot \ldots \cdot x_d)^n}{(n!)^d}.
$$
\end{prop}

We see $\GFNdk[d][d]$ as is a kind of generalised Bessel
function because
\[
    \GFNdk[2][2](x/2, -x/2) = J_0(x)
\]
where $J_\alpha$ is the classical Bessel function. This supports our feeling
that the general case leads to serious difficulties.

\subsection{Geometric interpretation} 
\label{geom_gnat}

As for non-ambiguous trees, we can give a geometric definition
of non-ambiguous trees of dimensions $(d,k)$ as follows.
We denote by $(e_1,\dots,e_d)$ the canonical basis of $\mathbb{R}^\d$ and
$(\X_1,\ldots,\X_\d)$ its dual basis, i.e. $\X_i$ is $\mathbb{R}$-linear
and $\X_i(e_i)=\delta_{i, j}$.
Let $P\in\mathbb{R}^\d$ and $\dir=\{i_1,\ldots,i_k\}$ a $(d,k)$-direction,
we call \emph{cone of origin $P$ and direction $\dir$} the set of points
$C(P, \dir):= \{P + a_1e_{i_1} + \cdots + a_ke_{i_k}\ |\ (a_1,\ldots,a_k)\in\mathbb{N}^k\}.$
\begin{defi}\label{def_gnat_geo}
A \emph{geometric non-ambiguous tree} of dimension $(\d,\k)$ and box $\gs$
is a non-empty set $\V$ of points of $\mathbb{N}^\d$ such that:
\begin{enumerate}[ref={\thecor.\arabic*}]
\item\label{gnat_geo_box} $\V$ is contained in $\i{1}{\wi_1}\times\cdots\times\i{1}{\wi_\d}$. 
\item\label{gnat_geo_root} $\V$ contains the point $(\wi_1,\ldots,\wi_\d)$,
    which is called the \emph{root}, 
\item\label{gnat_geo_cone} For $P\in \V$ different from the root, there exists
    a unique $(d,k)$-direction $\dir=\{i_1,\ldots,i_k\}$ such that the cone
    $c(P, \dir)$ contains at least one point different from $P$. We say that
    $P$ is of type $\dir$. 
\item\label{gnat_geo_distinct_and_interval} For each $i\in\i{1}{\d}$ and for each 
    $l\in\i{1}{\wi_i-1}$, the affine hyperplane $\{x_i=l\}$ contains exactly
    one point of type $\dir$ that contains $i$. 
\item\label{gnat_geo_affine} For $P$ and $P'$ two points of $\V$ belonging to a
    same affine space of direction $\Vect(e_{i_1},\ldots,e_{i_k})$, then,
    either $\forall j\in\i{1}{\k}\text{, }X_{i_j}(P)>X_{i_j}(P')$, or $\forall
    j\in\i{1}{\k}\text{, }X_{i_j}(P')>X_{i_j}(P).$
\end{enumerate}
\end{defi}

Let us compare the original definition (\cite{AvaBouBouSil14}) of non-ambiguous
trees recalled in Section~\ref{definitions} with Definition~\ref{def_gnat_geo}. Both
have a condition for the existence of a root \ref{condition_1_ana} and
\ref{gnat_geo_root}. The existence and uniqueness of a parent for a non-root
point are given by conditions \ref{condition_2_ana} and \ref{gnat_geo_cone},
moreover when $k\geq2$ we also need the condition \ref{gnat_geo_affine}.
Finally, the compactness is given by conditions \ref{condition_3_ana},
\ref{gnat_geo_box} and \ref{gnat_geo_distinct_and_interval}.

\begin{prop}
There is a simple bijection between the set of geometric non-ambiguous trees of
box $\gs$ and the set of non-ambiguous trees of
geometric size $\gs$.
\end{prop}
An example of the correspondence is given in Figure~\ref{fig_gnat}.
\begin{proof}
If $k=d$, $\V$ is of the form
\[
    \{(\wi,\dots,\wi),(\wi-1,\dots,\wi-1),\dots,(1,\dots,1)\},
\]
which corresponds exactly to non-ambiguous trees of dimension $(\d,\d)$ defined
in $\ref{def_gnat}$.

Let us now suppose that $k<d$.

 \textbf{Definition \ref{def_gnat} implies Definition \ref{def_gnat_geo}:}

Let $\N$ be a non-ambiguous tree of dimension $(d,k)$ as defined in
Definition \ref{def_gnat} and let $\gs$ be its geometric size.
The first step is to define the coordinates, also called completed label, of a vertex $U$
by replacing the $\bullet$ by integers in all the labels of vertices.
Let $U$ be a vertex of $\N$ such that its $i$th component is a $\bullet$. We replace the $i$th component of $U$
with the $i$th component of the first ancestor of $U$ with a $i$th component
different from $\bullet$. Such an ancestor exists
since the root has no $\bullet$ component.
As a consequence, using Point \ref{gnat_growth} of Definition \ref{def_gnat},  we have
for a vertex $V$ of completed label $(v_1,\dots,v_\d)$ that
if $V$ has a child $U$ indexed by a $(\d,\k)$-direction $\dir$
and of completed label $(u_1,\dots,u_\d)$, then for $i\in\dir$,
$v_i>u_i$ and for $i\not\in\dir$, $v_i=u_i$.
Moreover, \ref{gnat_distinct} and the definition of completed labels implies
that, for $U$ and $V$ two vertices, if there exists $i$ such that $u_i=v_i$,
then they have a common ancestor $W$ such that $w_i\neq\bullet$ and
$w_i=u_i=v_i$. We denote by $\mathfrak{a}_i(U)$, or equivalently
$\mathfrak{a}_i(V)$, the vertex $W$.

Let $V$ be a vertex of $\N$ of completed label $(u_1,\dots,u_\d)$.
We denote by $P_V\in\mathbb{N}^\d$ the point $(u_1,\dots,u_\d)$.
Let $\V$ be the set of points $\{P_V\ |\ V\text{ vertex of }\N\}$.
By definition, each vertex has a different completed label.
Let us prove that $\V$ satisfies the conditions of \ref{def_gnat_geo}.
\begin{enumerate}
\item It is a consequence of \ref{gnat_interval}, and Definition~\ref{gnat_gs}.
\item $\V$ contains $(\wi_1,\dots,\wi_\d)$, since $(\wi_1,\dots,\wi_\d)$
      is the label of the root of $\N$.
\item Let $P_U$ be a point of $\V$ different from $(\wi_1,\dots,\wi_\d)$. Since
    $U$ is not the root, it is a child indexed by a $(\d,\k)$-direction $\dir$,
    of a vertex $V$. Hence for $i\not\in\dir$, $\X_i(P_V)=\X_i(P_U)$ and for
    $i\in\dir$, $\X_i(P_V)>\X_i(P_U)$. So $P_V$ is in the cone of origin $P_U$
    and direction $\dir$, in particular, $P_U$ is of type $\dir$. Let us prove
    the uniqueness by contradiction. Suppose there is another $(\d,\k)$-direction
    $\dir'$ such that the cone of origin $P_U$ and direction $\dir'$ contains a
    point $P_{V'}$ different from $P_U$. Let $i\in\dir\setminus\dir'$. Then, by
    definition of the cone we have $X_i(P_U)=X_i(P_V')$. Since
    $\mathfrak{a}_i(U)=V$, then $\mathfrak{a}_i(V')=V$. Let
    $i'\in\dir'\setminus\dir$ then $X_i'(P_{V'})>X_i(P_U)=X_i(P_V)$, which is not
    possible since $V$ is an ancestor of $V'$. In particular the type of $P_U$
    corresponds to the type of $U$.
\item Let $i\in\llbracket1,d\rrbracket$ and $l\in\llbracket1,\wi-1\rrbracket$,
    by \ref{gnat_interval}, there exists $U$ such that $u_i\neq\bullet$ and
    $u_i=l$. Let $\dir$ be the index of $u$ then $i\in\dir$, hence, $P_U$
    satisfies \ref{gnat_geo_distinct_and_interval}. Suppose there exists
    another point $P_V$ satisfying \ref{gnat_geo_distinct_and_interval}, let
    $\dir'$ be its type. Then $V$ is indexed by $\dir'$ and $i\in\dir'$. Thus, $V$
    is another vertex such that $v_i\neq\bullet$ and $v_i=l$ which is in
    contradiction with \ref{gnat_distinct}.
\item Let $\F$ be an affine space of direction $\Vect(e_{i_1},\ldots,e_{i_k})$
    containing a point $P_U$. We denote by $\dir$ the set
    $\{e_{i_1},\ldots,e_{i_k}\}$. If $U$ is the root, then
    \ref{gnat_geo_affine} is satisfied. Else, let $(V_0,V_1,\dots,V_m)$ be the
    sequence of ancestors of $U$, i.e $V_0=U$ for all $j$, $V_j$ is a child of
    $V_{j+1}$ and $V_m$ is the root of $\N$. Let $l$ be the index such that
    $V_{l}$ is a child of $V_{l+1}$ not indexed by $\dir$ and $\forall
    j\in\i{0}{l-1}$, $V_j$ is the child indexed by $\dir$ of $\V_{j+1}$, $V_l$
    will be denoted $V$. Let $\dir'$ be the $(d,k)$-direction indexing $V$ and
    let $i\in\dir'\setminus\dir$. If $V$ is the root then $\dir'=\llbracket 1,d
    \rrbracket$. Let $P_{U'}$ be another point of $\F$, since
    $i\not\in\dir$, then $u'_i=u_i$. Hence, we have
    $\mathfrak{a}_i(U')=\mathfrak{a}_i(U)=V$.
    Since for all $j\not\in\dir$, $u'_j=v_j$, the path from $U'$ to $V$ contains
    only vertices indexed by $\dir$. Hence, by definition of $V$, $U$ is an
    ancestor of $U'$ or the converse, which proves \ref{gnat_geo_affine}.
\end{enumerate}

\textbf{Definition \ref{def_gnat_geo} implies Definition \ref{def_gnat}:}
Let $\V$ be a non-ambiguous tree defined with Definition \ref{def_gnat_geo}. We
start by constructing the underlying $\binom{d}{k}$-ary tree $\M$ of $\V$. The
vertices of $\M$ correspond to the points of $\V$. In particular, the root of
$\M$ corresponds to the root of $\V$. Let $P$ be a point of $\V$, we denote by
$V_P$ the corresponding vertex of $\M$. Let $P$ be a point of $\V$ different
from the root, let $\dir=\{e_{i_1},\ldots,e_{i_k}\}$ be the $(d,k)$-direction
defined by \ref{gnat_geo_cone}. Using \ref{gnat_geo_affine}, we can define
without ambiguity the parent of $V_P$ as the vertex $V_{P'}$ such that $P'$ is the
closest point to $P$ belonging to $c(P,\dir)$. $V_P$ is the child of $V_{P'}$
indexed by $\dir$, moreover, for all $i\in\dir$, $X_i(P)<X_i(P')$ and for all
$i\not\in\dir$, $X_i(P)=X_i(P')$. The labelling is done as follows. For each
point $P$ of $\V$ we label the vertex $V_P$ with the coordinates of $P$. Then, given a vertex
$V$ of type $\dir$, for all $i\not\in\dir$, we replace the $i$th
component of its label with $\bullet$. Thus, if the $i$th component of a vertex
$V_{P'}$ is equal to $l$, then for each descendant $V_P$ of $V_{P'}$ we have
$\X_i(P)\leqslant l$ and if $\X_i(P)=l$ then the $i$th component of $V_{P'}$
is $\bullet$. Let us prove that the conditions of Definition~\ref{def_gnat} are
satisfied.
\begin{enumerate}
\item By construction of the labels.
\item Let $P$ and $P''$ be two points such that $V_P$ is a descendant of
    $V_{P''}$ indexed by $\dir$ such that $i$ belongs to $\dir$. Let $P'$ be
    the father of $P$, then $\X_i(P)<\X_i(P')\leqslant \X_i(P'')$.
\item Let $V_P$ be a vertex of $\M$ such that its $i$th component is different
    from $\bullet$. If $V_P$ is the root, then all the vertices of $\M$ are its
    descendants, hence its $i$th label appears only once. Else, $V_P$ is a
    child indexed by $\dir$ of a vertex $V_{P'}$. In particular, $\dir$
    contains $i$ since the $i$th component of $V_P$ is not $\bullet$. Hence, by
    \ref{gnat_geo_distinct_and_interval}, the $i$th component of $V_P$ is
    unique.
\item \ref{gnat_geo_distinct_and_interval} implies that for each
    $i\in\i{1}{\d}$, for all $l\in\i{1}{\wi_i-1}$ there is a point $P$ of type
    $\dir$ such that $\X_i(P)=l$ and $i\in\dir$, so that the $i$th component of
    $V_P$ is equal to $l$. Moreover, the coordinates of the root are
    $(\wi_1,\dots,\wi_\d)$ and $\V$ is contained in the box
    $\i{1}{\wi_1}\times\cdots\times\i{1}{\wi_\d}$. Therefore, the set of $i$th
    components, different from $\bullet$, is the interval $\i{1}{\wi_i}$.
\end{enumerate}
\end{proof}

\section{A new statistic on binary trees: the hook statistic}\label{hook} 

We present in this section a bijection between binary trees and ordered trees, sending the vertices to edges and the hook statistic defined in Definition \ref{def_hook} to the number of vertices having at least a child which is a leaf, what we will call the \emph{child-leaf statistic}.
The corresponding integer series appears as \cite[A127157]{oeis} in OEIS.

We denote by $\Bp$ and $\Op$ respectively the exponential generating series of
these trees, with these statistics, the variable $x$ indexing the number of
vertices in $\Bp$ and the number of edges in $\Op$, and $t$ the statistic.
Then, these generating series satisfy:
\begin{prop}
The generating series of binary trees with hook statistic and ordered trees with the child-leaf statistic are given by the following functional equations:
\begin{equation*}
\Bp=  1+xt\times\left(\frac{1}{1-x\Bp} \right)^2
\end{equation*}
\begin{equation*}
\Op=\frac{1}{1-x(\Op-1))}\times\left(1+xt\times\frac{1}{1-x\Op} \right)
\end{equation*}
These generating series are equal.
\end{prop} 

\begin{proof}
The first functional equation is obtained by considering the vertices in the hook of the root: there can be none or there is a root, a list of left descendant (whose right child is a binary tree) and a list of right descendant (whose left child is a binary tree).

The second functional equation is obtained by considering, if the ordered tree
is not reduced to a vertex, the first leaf of the root from left to right, if
it exists. Then, on the left side of this leaf, there is a list of
ordered trees not reduced to a vertex and on the right side a list of ordered trees, if there is a leaf.
Then, by multiplying the preceding equations by $1-x\Bp$ and $1-x(\Op-1))$ respectively, they are equivalent to:
\begin{equation*}
\Bp -x\Bp^2=  1-x\Bp + xt \frac{1}{1-x\Bp}
\end{equation*}
\begin{equation*}
\Op -x \Op^2 +x \Op = 1 +xt\frac{1}{1-x\Op}.
\end{equation*}
\end{proof}

Let us now exhibit a bijection between these two objects. This bijection comes from the following equation:
\begin{equation*}
(\Bp-1) -x (\Bp-1)-x (\Bp-1)^2 = xt \frac{1}{1-x\Bp}.
\end{equation*}

This equation can be viewed as considering only binary trees whose root has no left descendants or ordered trees such that the leftmost child of the root is a leaf.
We obtain the following bijection:
\begin{prop}
The map $\zeta$ sends a binary tree $\B$ to an ordered tree $\O$ by mapping:
\begin{itemize}
\item the leftmost descendant of the root, if it is a leaf, to an edge between the root and its only child
\item the leftmost descendant of the root $v$ to an edge between the root of the tree associated with the descendants of $v$ and the root of the tree obtained from what is left
\item the set of right descendants of the root to the set of children of the root.
\end{itemize}
It is a bijection between binary trees and ordered trees, sending the vertices to edges and the hook statistic to the child-leaf statistic.
\end{prop}
The bijection is described in Figure~\ref{zeta} and the explicit correspondance
for the sizes one, two and three is given in Table~\ref{tableau_zeta} the first
terms in the bijection. Another way of describing recursively the bijection
$\zeta$ is given in Figure~\ref{zeta2}, the empty binary tree is still send to
the ordered tree reduced to one vertex.
\begin{figure}
\scalebox{0.75}{
\begin{tikzpicture}
\draw (-0.5,0) -- (0,0.75) -- (0.5,0) -- cycle;
\draw[fill] (0,0.75) circle(0.1);
\draw (-0.5,0)--(-1,-0.75);
\draw[fill] (-1,-0.75) circle(0.1);
\draw (1,0) node{$\Leftrightarrow$};
\draw (1,-0.75) -- (1.5,0) -- (2,-0.75) -- cycle;
\draw[fill] (1.5,0) circle(0.1);
\draw (1.5,0)--(1.5,0.75);
\draw[fill] (1.5,0.75) circle(0.1);
\end{tikzpicture} }
\hspace{1cm}
\scalebox{0.75}{
\begin{tikzpicture}
\draw (-0.5,0) -- (0,0.75) -- (0.5,0) -- cycle;
\draw (0,0.25) node{$1$};
\draw[fill] (0,0.75) circle(0.1);
\draw (-0.5,0)--(-1,-0.75);
\draw[fill] (-1,-0.75) circle(0.1);
\draw (-1,-2.25) -- (-0.5,-1.5) -- (0,-2.25) -- cycle;
\draw (-0.5,-2) node{$2$};
\draw[fill] (-0.5,-1.5) circle(0.1);
\draw (-0.5,-1.5)--(-1,-0.75);
\draw (1,-0.75) node{$\Leftrightarrow$};
\draw (2,-1) -- (2.5,-0.25) -- (3,-1) -- cycle;
\draw (2.5,-0.75) node{$2$};
\draw[fill] (2.5,-0.25) circle(0.1);
\draw (2.5,-0.25)--(1.5,-0.75);
\draw[fill] (1.5,-0.75) circle(0.1);
\draw (1,-1.5) -- (1.5,-0.75) -- (2,-1.5) -- cycle;
\draw (1.5,-1.25) node{$1$};
\end{tikzpicture} }
\hspace{1cm}
\scalebox{0.75}{
\begin{tikzpicture}
\draw (0,2.5)--(2.25,1);
\draw[fill] (0,2.5) circle(0.1);
\draw[fill] (0.75,2) circle(0.1);
\draw[fill] (1.5,1.5) circle(0.1);
\draw[fill] (9/4,1) circle(0.1);
\draw (0.75,1.5) node{$1$};
\draw (1.5,1) node{$2$};
\draw (9/4,0.5) node{$3$};
\draw (0.30,1.25) -- (0.75,2) -- (1.20,1.25) -- cycle;
\draw (1.05,0.75) -- (1.5,1.5) -- (1.95,0.75) -- cycle;
\draw (1.80,0.25) -- (2.25,1) -- (2.70,0.25) -- cycle;
\draw (3,1.25) node{$\Leftrightarrow$};
\draw[fill] (4,2) circle(0.1);
\draw[fill] (3.5,1) circle(0.1);
\draw[fill] (4.5,1) circle(0.1);
\draw[fill] (5.5,1) circle(0.1);
\draw[fill] (6.5,1) circle(0.1);
\draw (4.5,0.5) node{$1$};
\draw (5.5,0.5) node{$2$};
\draw (6.5,0.5) node{$3$};
\draw (4,2)--(3.5,1);
\draw (4,2)--(4.5,1);
\draw (4,2)--(5.5,1);
\draw (4,2)--(6.5,1);
\draw (4.05,0.25) -- (4.5,1) -- (4.95,0.25) -- cycle;
\draw (5.05,0.25) -- (5.5,1) -- (5.95,0.25) -- cycle;
\draw (6.05,0.25) -- (6.5,1) -- (6.95,0.25) -- cycle;
\end{tikzpicture}} 
 \caption{Bijection $\zeta$}\label{zeta}
 \end{figure}
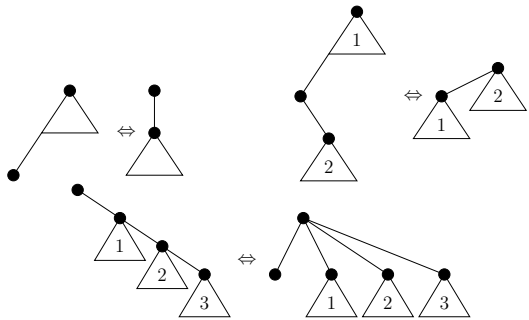
 \begin{table}
 \begin{tabular}{|c|cccccccc|}
 \hline
 binary trees & 
\begin{tikzpicture}
\node[blue] (r) at (0,0) {$\bullet$};
 \end{tikzpicture}&
\begin{tikzpicture}
\node[blue] (r) at (0,0) {$\bullet$};
\node[blue](fg) at (-0.5,-0.75){$\bullet$};
\draw[blue] (r) --(fg);
\end{tikzpicture}&
\begin{tikzpicture}
\node[blue] (r) at (0,0) {$\bullet$};
\node[blue](fd) at (0.5,-0.75){$\bullet$};
\draw[blue] (r) --(fd);
\end{tikzpicture}&
\begin{tikzpicture}
\node[blue] (r) at (0,0) {$\bullet$};
\node[blue](fg) at (-0.5,-0.75){$\bullet$};
\node[blue](fd) at (0.5,-0.75){$\bullet$};
\draw[blue] (fd)--(r) --(fg);
\end{tikzpicture} &
\begin{tikzpicture}
\node[blue] (r) at (0,0) {$\bullet$};
\node[blue](fg) at (-0.5,-0.75){$\bullet$};
\node[blue](fgg) at (-1,-2*0.75){$\bullet$};
\draw[blue] (r) --(fg)--(fgg);
\end{tikzpicture}&
\begin{tikzpicture}
\node[blue] (r) at (0,0) {$\bullet$};
\node[blue](fd) at (0.5,-0.75){$\bullet$};
\node[blue](fdd) at (1,-2*0.75){$\bullet$};
\draw[blue] (r) --(fd)--(fdd);
\end{tikzpicture}&
\begin{tikzpicture}
\node[blue] (r) at (0,0) {$\bullet$};
\node[blue](fg) at (-0.5,-0.75){$\bullet$};
\node[purple](fgg) at (0,-2*0.75){$\bullet$};
\draw[blue] (r) --(fg);
\draw (fg)--(fgg);
\end{tikzpicture}&
\begin{tikzpicture}
\node[blue] (r) at (0,0) {$\bullet$};
\node[blue](fd) at (0.5,-0.75){$\bullet$};
\node[purple](fdd) at (0,-2*0.75){$\bullet$};
\draw[blue] (r) --(fd);
\draw (fd)--(fdd);
\end{tikzpicture}
\\ \hline
 ordered trees &
\begin{tikzpicture}
 \node[blue](r) at (0,0){$\bullet$};
\node(f) at (0,-0.75){$\bullet$};
 \draw (r)--(f);
 \end{tikzpicture}&
\begin{tikzpicture}
 \node(r) at (0,0){$\bullet$};
\node[blue](f) at (0,-0.75){$\bullet$};
\node(ff) at (0,-2*0.75){$\bullet$};
 \draw (r)--(f)--(ff);
 \end{tikzpicture}&
\begin{tikzpicture}
\node[blue] (r) at (0,0) {$\bullet$};
\node(fg) at (-0.5,-0.75){$\bullet$};
\node(fd) at (0.5,-0.75){$\bullet$};
\draw (fd)--(r) --(fg);
\end{tikzpicture} &
\begin{tikzpicture}
\node(rr) at (0,0.75){$\bullet$};
\node[blue] (r) at (0,0) {$\bullet$};
\node(fg) at (-0.5,-0.75){$\bullet$};
\node(fd) at (0.5,-0.75){$\bullet$};
\draw (fd)--(r) --(fg);
\draw (r)--(rr);
\end{tikzpicture} &
  \begin{tikzpicture}
 \node(r) at (0,0){$\bullet$};
\node(f) at (0,-0.75){$\bullet$};
\node[blue](ff) at (0,-2*0.75){$\bullet$};
\node(fff) at (0,-3*0.75){$\bullet$};
 \draw (r)--(f)--(ff)--(fff);
 \end{tikzpicture}&
\begin{tikzpicture}
\node[blue] (r) at (0,0) {$\bullet$};
\node(fg) at (-0.5,-0.75){$\bullet$};
\node(fd) at (0.5,-0.75){$\bullet$};
\node(fm) at (0,-0.75){$\bullet$};
\draw (fd)--(r) --(fg);
\draw (r)--(fm);
\end{tikzpicture} &
\begin{tikzpicture}
\node[blue] (r) at (0,0) {$\bullet$};
\node(fg) at (-0.5,-0.75){$\bullet$};
\node[purple](fd) at (0.5,-0.75){$\bullet$};
\node(f) at (0.5,-2*0.75){$\bullet$};
\draw (f)--(fd)--(r) --(fg);
\end{tikzpicture} &
\begin{tikzpicture}
\node[blue] (r) at (0,0) {$\bullet$};
\node[purple] (fg) at (-0.5,-0.75){$\bullet$};
\node(f) at (-0.5,-2*0.75){$\bullet$};
\node(fd) at (0.5,-0.75){$\bullet$};
\draw (fd)--(r) --(fg)--(f);
\end{tikzpicture}  
 \\ \hline
  \end{tabular}
\caption{First terms of the bijection $\zeta$}\label{tableau_zeta} 
\end{table}
\begin{figure}
\scalebox{1}{
    \begin{tikzpicture}
        \draw[blue] (0,2.5)--++ (-3,-1.5);
        \draw[blue] (0,2.5)--++ (3,-1.5);
        \draw (-3,1) --++ (.5,-.75);
        \draw (-2,1.5) --++ (.5,-.75);
        \draw (-1,2) --++ (.5,-.75);
        \draw (1,2) --++ (-.5,-.75);
        \draw (2,1.5) --++ (-.5,-.75);
        \draw (3,1) --++ (-.5,-.75);
        \draw[fill,blue] (-3,1) circle(0.1);
        \draw[fill,blue] (-2,1.5) circle(0.1);
        \draw[fill,blue] (-1,2) circle(0.1);
        \draw[fill,blue] (0,2.5) circle(0.1);
        \draw[fill,blue] (1,2) circle(0.1);
        \draw[fill,blue] (2,1.5) circle(0.1);
        \draw[fill,blue] (3,1) circle(0.1);
        \draw (-2.5,-.25) node{$a$};
        \draw (-1.5,.25) node{$b$};
        \draw (-.5,.75) node{$c$};
        \draw (.5,.75) node{$1$};
        \draw (1.5,.25) node{$2$};
        \draw (2.5,-.25) node{$3$};
        \draw (-2.5,.25)--++ (-.45,-.75) --++ (.9,0) -- cycle;
        \draw (-1.5,.75)--++ (-.45,-.75) --++ (.9,0) -- cycle;
        \draw (-.5,1.25)--++ (-.45,-.75) --++ (.9,0) -- cycle;
        \draw (.5,1.25)--++ (-.45,-.75) --++ (.9,0) -- cycle;
        \draw (1.5,0.75)--++ (-.45,-.75) --++ (.9,0) -- cycle;
        \draw (2.5,0.25)--++ (-.45,-.75) --++ (.9,0) -- cycle;
        \draw (4.5,1.25) node{$\Longleftrightarrow$};
    \end{tikzpicture}
    \hspace{5mm}
    \begin{tikzpicture}
        \draw (4,2)--++(2.25,2.25);
        \draw (6.25,3.75) node{\Small$\zeta(a)$};
        \draw (5.5,3) node{\Small$\zeta(b)$};
        \draw (4.75,2.25) node{\Small$\zeta(c)$};
        \draw (4.5,0.5) node{\Small$\zeta(1)$};
        \draw (5.5,0.5) node{\Small$\zeta(2)$};
        \draw (6.5,0.5) node{\Small$\zeta(3)$};
        \draw (6.25,4.25) --++ (-.45,-.75) --++ (.9,0) -- cycle;
        \draw (5.5,3.5) --++ (-.45,-.75) --++ (.9,0) -- cycle;
        \draw (4.75,2.75) --++ (-.45,-.75) --++ (.9,0) -- cycle;
        \draw (4,2)--(3.5,1);
        \draw (4,2)--(4.5,1);
        \draw (4,2)--(5.5,1);
        \draw (4,2)--(6.5,1);
        \draw[fill] (6.25,4.25) circle(0.1);
        \draw[fill] (5.5,3.5) circle(0.1);
        \draw[fill] (4.75,2.75) circle(0.1);
        \draw[fill,blue] (4,2) circle(0.1);
        \draw[fill] (3.5,1) circle(0.1);
        \draw[fill] (4.5,1) circle(0.1);
        \draw[fill] (5.5,1) circle(0.1);
        \draw[fill] (6.5,1) circle(0.1);
        \draw (4.05,0.25) -- (4.5,1) -- (4.95,0.25) -- cycle;
        \draw (5.05,0.25) -- (5.5,1) -- (5.95,0.25) -- cycle;
        \draw (6.05,0.25) -- (6.5,1) -- (6.95,0.25) -- cycle;
    \end{tikzpicture}
}
\caption{Bijection $\zeta$, alternative description.}\label{zeta2}
\end{figure}
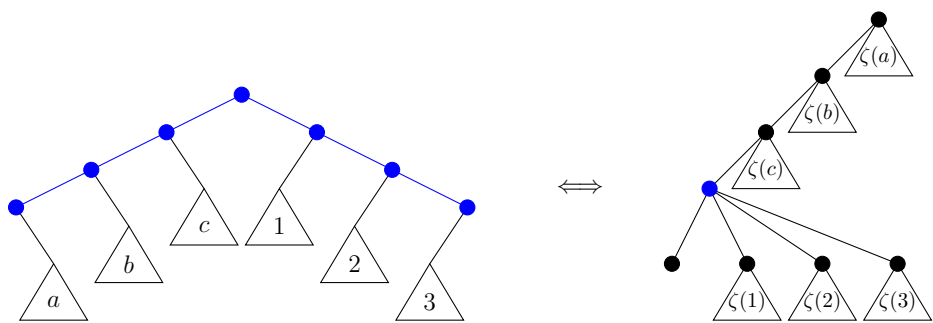

\section*{Perspectives}

In this work, we give new results about NATs, and we generalise the
definition of NATs to higher dimensions with a choice on the dimension of the
edges. It gives rise to several questions that we now detail.

We exhibit nice formulas for the generating function and the generating polynomial
of non-ambiguous trees that take into account the number of points in the first
column and in the first row. Those two parameters correspond to the parameters
$\alpha$ and $\beta$ of the PASEP. These formulas have been obtained with two different
technics: by solving a differential equation and by decomposing non-ambiguous
trees into hooks. In the context of the PASEP, it raises natural questions. Is it
possible to introduce the parameter $q$ of the PASEP in either one of them ? What
 can we deduce from the hook decomposition of tree-like tableaux ?

We give a polynomial analogue for the hook formula enumerating the NATs with a
fixed binary tree, by adapting the methods of \cite{HivNovThib08}. Since there is a hook formula also for NATs of higher dimension, we could extend the work
to higher dimensions to get polynomial analogue.

The generating functions of generalised NATs satisfy differential equations
similar to the case of NATs. While we give a solution for the case of NATs and
get a nice closed form, we have not been able to tackle the general case. It
would be interesting to find a generic way for solving this type of
differential equations.

By generalising the NATs to higher dimension, we answer a question raised
in the perspectives of \cite{AvaBouBouSil14}. In this last paper, the authors also study NATs with a complete underlying binary tree, and
 obtain nice combinatorial identities. It would be interesting to see if
the same happens in higher dimension.

Finally, as mentioned in the introduction, NATs correspond to tree-like
tableaux of rectangular shape. The question of the generalisation of tree-like tableaux to
higher dimension is raised by our work and we hope we will obtain again an insertion algorithm, which is a key
property of tree-like tableaux.

\ \\
{\bf Acknowledgement.}
The authors thank Samanta Socci for fruitful discussions which were
the starting point of the generalisation of non-ambiguous trees.
This research was driven by computer exploration using the open-source software
\texttt{Sage}~\cite{sage} and its algebraic combinatorics features developed by
the \texttt{Sage-Combinat} community~\cite{Sage-Combinat}.
The authors also thank the referees for their helpful suggestions.

\nocite{*}
\bibliographystyle{alpha}
\bibliography{biblio}
\end{document}